\documentclass[11pt,a4paper,reqno,intlimits,sumlimits]{amsart}

\usepackage{amsmath,amsthm,amsfonts,amssymb}
\usepackage{ifpdf}
\usepackage{color}
\usepackage{graphicx}
\usepackage{epsfig}
\usepackage{caption}
\usepackage{booktabs}
\usepackage{todonotes}
\usepackage{xspace}
\usepackage{enumitem}
\usepackage[round,sort&compress,longnamesfirst]{natbib}
\usepackage{times}

\usepackage[left=1.2in,right=1.2in,top=1.2in,bottom=1.1in]{geometry}

\usepackage[bookmarksopen,pdfstartview=FitH]{hyperref}
\hypersetup{colorlinks,%
            citecolor=blue,%
            filecolor=blue,%
            linkcolor=blue,%
            urlcolor=blue}%
            
\theoremstyle{plain}
\newtheorem{theorem}{Theorem}[section]
\newtheorem{lemma}[theorem]{Lemma}

\theoremstyle{definition}
\newtheorem{definition}[theorem]{Definition}
\newtheorem{remark}[theorem]{Remark}

\newtheorem*{ASRn}{Assumption ($\mathbb{R}_n$)}
\newtheorem*{ASS1}{Assumption ($\mathbb{S}_1$)}
\newtheorem*{ASS2}{Assumption ($\mathbb{S}_2$)}
\newtheorem*{ASR}{Assumption ($\mathbb{RHO}$)}
\newtheorem*{ASE}{Assumption ($\mathbb{ELL}$)}

\newcommand{\E}{{\mathbb{E}}}
\newcommand{\N}{{\mathbb{N}}}

\newcommand{\Q}{{\mathbb{Q}}}
\newcommand{\R}{{\mathbb{R}}}
\newcommand{\D}{{\mathbb{D}}}
\newcommand{\F}{\mathcal{F}}

\newcommand{\I}{\mathcal{I}}

\makeatletter
\DeclareRobustCommand*\cal{\@fontswitch\relax\mathcal}
\makeatother

\newcommand{\be}{\begin{equation}}
\newcommand{\ee}{\end{equation}}
\newcommand{\bea}{\begin{eqnarray}}
\newcommand{\eea}{\end{eqnarray}}
\newcommand{\beast}{\begin{eqnarray*}}
\newcommand{\eeast}{\end{eqnarray*}}
\newcommand{\bproof}{\begin{proof}}
\newcommand{\eproof}{\end{proof}}

\numberwithin{equation}{section}

\numberwithin{table}{section}

\def\e{\mathrm{e}}
\def\ud{\ensuremath{\mathrm{d}}}
\def\dt{\ud t}
\def\ds{\ud s}
\def\dx{\ud x}

\def\dW{\ud W}

\def\lib{LIBOR\xspace}

\begin{document}

\title[Expansion formulas for European quanto options]{Expansion formulas for European quanto options\\in a local volatility FX-LIBOR  model}
\author[J. Hok]{Julien Hok}
\author[P. Ngare]{Philip Ngare}
\author[A. Papapantoleon]{Antonis Papapantoleon}

\address{Credit Agricole CIB, Broadwalk House, 5 Appold St, EC2A 2DA London, United Kingdom}
\email{julienhok@yahoo.fr}

\address{School of Mathematics, University of Nairobi, PO Box 30197-0010, Nairobi, Kenya}
\email{pngare@uonbi.ac.ke}

\address{Department of Mathematics, National Technical University of Athens, Zografou Campus, 15780 Athens, Greece}
\email{papapan@math.ntua.gr}


\date{} \frenchspacing



\begin{abstract} 
We develop an expansion approach for the pricing of European quanto options written on LIBOR rates (of a foreign currency). 
We derive the dynamics of the system of foreign LIBOR rates under the domestic forward measure and then consider the price of the quanto option. 
In order to take the skew/smile effect observed in fixed income and FX markets into account, we consider local volatility models for both the LIBOR and the FX rate. 
Because of the structure of the local volatility function, a closed form solution for quanto option prices does not exist. 
Using expansions around a proxy related to log-normal dynamics, we derive approximation formulas of Black--Scholes type for the price, that have the benefit of giving very rapid numerical procedures.  
Our expansion formulas have the major advantage that they allow for an accurate estimation of the error, using Malliavin calculus, which is directly related to the maturity of the option, the payoff, and the level and curvature of the local volatility function.
These expansions also illustrate the impact of the quanto drift adjustment, while the numerical experiments show an excellent accuracy.
\end{abstract}

\maketitle

\section{Introduction}
\label{sec:introduction}

We are interested in the pricing of European quanto options on LIBOR rates. 
These correspond to a type of derivative in which the underlying rate is denominated in one currency (foreign currency) but the payment is made in another currency (domestic currency). 
Such products are attractive for speculators and investors who wish to have exposure to a foreign asset, but without the corresponding exchange rate risk. 
Think, for instance, of a Euro-based investor who is seeking exposure on the GBP LIBOR rate, but does not want to be exposed to changes of the GBP/EUR foreign exchange rate. 
A European quanto option on the GBP LIBOR rate is a very suitable financial product for her, as it has the payoff of a standard non-quanto option on the GBP LIBOR rate and converts the payout with a guaranteed rate of 1 from GBP into Euro at maturity.  

In an arbitrage-free framework, the pricing of quanto options  can be performed under the domestic forward measure. 
In order to express the dynamics of the underlying LIBOR rate under this pricing measure, one has to apply Girsanov's theorem, leading to a drift term which depends on the volatility of the LIBOR rate, on the volatility of the FX rate, and on the correlation between the LIBOR and the FX rate. 
This drift term leads to an adjustment in the pricing that is referred to as \textit{quanto adjustment} and falls into the more general category of what is called in mathematical finance \textit{convexity adjustment}. 

This class of contracts, termed \textit{exotic European options} by \citet{Pelsser00}, are widely traded over the counter (OTC).
The correct pricing and risk management of European quanto options constitutes an important issue in the financial industry. 
The consideration of the market skew/smile for interest rates and FX rates is fundamental for a correct valuation of European quanto derivatives as discussed in \citet{Romo12}.   
In reference textbooks and articles, see \textit{e.g.} \citet{MusielaRut05}, \citet{BrigiMer06} or \citet{Reiner92}, a simplified Black--Scholes model is considered in order to obtain analytical formulas. 
A similar practical approach is commonly used in the financial industry; see Section \ref{commonmktappro} or \textit{e.g.} \citet{Romo12} and \citet{ChrisJac04} for more details. 
However, it does not take into account properly the skew/smile effects of the underlying assets in the quanto drift adjustment. 
These issues with the commonly used approach and the importance of incorporating the skew/smile properly in the valuation of European quanto options are studied and discussed in \textit{e.g.} \citet{Romo12}, \citet{Jackel10}, \citet{Giese12} or \citet{VongMateo14}.     

Local volatility models, either parametric or non-parametric, see \textit{e.g.} \citet{Dupire94,DerKa98,Rubi94,Jac08} or \citet{Cox75}, usually capture the surface of implied volatilities more precisely than other approaches, such as stochastic volatility models; see \textit{e.g.} \citet{MadQianRen07} or \citet{Romo12} for discussions. 
Moreover, the findings in \citet{Romo12} or \citet{HullSuo02} indicate that the local volatility model can be a correct approach to price European quanto derivatives in the presence of volatility skew/smile. 

Motivated by the discussions above, we propose to evaluate European quanto options in a general local volatility framework.  
Because of its generality, it is often difficult to get analytical formulas for pricing, especially in a high-dimensional case. 
In general, the effective pricing requires the use of a numerical method, based either on PDE (partial differential equation) techniques or Monte Carlo simulations, which can be prohibitively time-consuming for real-time applications. 
Only in very few cases does one have closed-form formulas; \textit{cf.} \citet{ACCL01}. 
In the case of homogeneous volatility, singular perturbation techniques (\citet{HaganWood99}) have been used to obtain asymptotic expressions for the price of vanilla options (call, put). 
Implied volatility formulas are derived using asymptotic expansion methods for short maturities, as in \citet{BeresBusFlorent02}, \citet{PHLabordere05} and \citet{ACCL01}. 
In a more general diffusion setting, approximations of the density function and option prices are derived based on the small disturbance asymptotics, see \textit{e.g.} \citet{KumiTaka92,Yoshi92b} or \citet{Taka95,Taka99}  or \cite{Taka09} for a review. 
By adapting the singular perturbation method used in \citet{HaganWood99}, several authors have developed expansion formulas for the density function of the underlying process and option prices in a general local volatility model; see \textit{e.g.} \citet{PagPas12} and \citet{FoPagPas13}.     

The purpose of the present article is to provide simple and accurate approximation formulas for quanto options in a general local volatility model.
Towards this end, we apply the perturbation method using a proxy introduced in \citet{MiriGobetBen02}.
This method has been applied and extended in many directions, see \textit{e.g.} \citet{MiriGobetBen00,MiriGobetBenIR02}, \citet{GobetMiri14}, \citet{GobetHok14} and \citet{GobetBompis14}. 
We derive expansion formulas, which are of Black--Scholes type, and develop the analysis using Malliavin calculus, to provide accurate estimations of the errors.  
We believe that rigorous error estimates are of prime importance because the accuracy of our expansion formulas depends on the regularity of the payoff function. Once done, this brings confidence in the derived expansions and sheds light on the needed assumptions; see our main results in Theorems \ref{ApproxFormulaProbaApproa2order} and \ref{ApproxFormulaProbaApproa3order}.

A major advantage of our expansion formulas is that they clearly illustrate the impact of the quanto drift adjustment and provide very rapid numerical procedures for its implementation.
The numerical tests, see Section \ref{NumericalExperiments}, show that our formulas constitute a very accurate approximation. 
Our interest in such problems was motivated by specific applications to European quanto derivatives on LIBOR rates, hence we specialize our study to that setting.
However, the approximation methodology and results could be applied to other financial assets as well. 
Moreover, we focus on single-curve LIBOR models as they constitute the basis for multi-curve models.
The extension to multiple curves is straightforward given the analysis and results of the present paper, however it is also very tedious.

This article is structured as follows: Section \ref{sec:model-setup} introduces local volatility models for simultaneously modeling FX and LIBOR rates, as well as quanto options.
Section \ref{sec:expa-quanto} outlines the approach to quanto pricing via expansions around a proxy model and states the main results, which are second and third order expansions for the prices of quanto options.
Section \ref{section:Proofs} provides an error analysis and the derivation of the second order expansion formula, while Section \ref{NumericalExperiments} provides numerical results.
Finally, the Appendices contains some auxiliary results and the derivation of the third order expansion formula.





\section{FX-LIBOR models and European quanto options}
\label{sec:model-setup}

\subsection{A local volatility FX-LIBOR framework}

Let $(\Omega,\mathcal{F},\mathbb{F},\mathbb{Q}_N)$ denote a filtered probability space where the filtration $\mathbb{F}=(\F_t)_{t\in{[0,T_N]}}$ satisfies the 
usual conditions and $T_N$ denotes a finite time horizon. 
Let also $\mathcal{T}=\{0=T_0<T_1<\dots<T_N\}$ denote a discrete tenor structure where $\delta_i=T_{i+1}-T_i$ is the accrual fraction for the period $[T_i,T_{i+1}]$, and define $\mathcal{I}=\{1,\dots,N\}$. 
The dates $(T_i)_{i\in\I}$ correspond to maturity dates of traded instruments.

We assume the existence of an arbitrage-free system of domestic and foreign zero coupon bonds, denoted respectively $(B(\cdot,T_i))_{i\in\mathcal{I}}$ and 
$(B^f(\cdot,T_i))_{i\in\mathcal{I}}$. 
We further consider domestic and foreign forward martingale measures, denoted by $(\Q_i)_{i\in\mathcal{I}}$ and $(\Q_i^f)_{i\in\mathcal{I}}$, where the corresponding zero coupon bond acts as the numeraire for each forward measure. 
Let $W=(W_t)_{t\in{[0,T_N]}}$ and $W^f=(W^f_t)_{t\in{[0,T_N]}}$ denote standard $d$-dimensional Brownian motions relative to the domestic and foreign terminal forward measures $\Q_N$ and $\Q_N^f$ respectively.

Let $(L_i)_{i\in\I}$ and $(L^f_i)_{i\in\I}$ denote the domestic and foreign forward \lib rates, \textit{i.e.} the discretely compounded forward rates for investing in the time period $[T_i,T_{i+1}]$ in the domestic and foreign market. 
Their relation to zero coupon bonds is classically given by
\begin{align}
1 + \delta_i L_i(t) =  \frac{B(t,T_i)}{B(t,T_{i+1})}
  \ \ \text{ and } \ \
1 + \delta_i L_i^f(t) =  \frac{B^f(t,T_i)}{B^f(t,T_{i+1})}.
\end{align}
The dynamics of the system of foreign LIBOR rates $(L_i^f)_{i\in\I}$ is provided by a local volatility model of the form
\begin{align}\label{equ1a}
\ud L^f_i(t) &= L^f_i(t) \lambda_i\big(t,L^f_i(t)\big) \ud W^{f, i+1}_t,
\end{align}
where the $\Q^f_{i+1}$-Brownian motion is related to the terminal Brownian motion via
\begin{equation}\label{eq:for-ter-BM}
W^{f,i+1} 
  = W^{f,N} - \sum_{k=i+1}^{N} \int_0^\cdot 
    \frac{\delta_k L^f_k(t)}{1+\delta_k L^f_k(t)} 
    \lambda_k\big(t,L^f_k(t)\big) \dt,
\end{equation} 
for all $t\in[0,T_{i+1}]$. 
The functions $\lambda_i:[0,T_i]\times \R\to\R^d_+$, $i\in\I$, are continuous, deterministic and satisfy a suitable linear growth condition, \textit{cf.} \citet[\S10.3]{BrigiMer06}. 
They represent the local volatility of the foreign forward LIBOR rate $L_i^f$.

Let $(X(t))_{t\in[0,T_N]}$ denote the foreign exchange (FX) rate expressed in terms of units of domestic currency per unit of foreign currency. 
The FX forward rate for settlement at time $T_i$, denoted by $(X_i(t))_{t\in[0,T_i]}$, is defined by no-arbitrage arguments and provided by 
\begin{equation}\label{equ4}
 X_i(t) = \frac{B^f(t,T_i) X(t)}{B(t,T_i)},
\end{equation}
for all $t\in[0,T_i]$. 
The FX forward rate is, per definition, a $\Q_i$-martingale, and we assume it follows again a local volatility model of the form
\begin{equation}\label{equ5}
 \ud X_i(t)=  X_i(t) \sigma_i\big(t,X_i(t)\big) \ud W^i_t,
\end{equation} 
where $W^i$ is the $\Q_i$-Brownian motion and $\sigma_i:[0,T_i]\times \R\to\R^d_+$, $i\in\I$, is a continuous, deterministic function satisfying a suitable linear growth condition, and represents the local volatility of the FX forward rate. 

The domestic and foreign forward Brownian motions are related via
\begin{align}
W^{f,i} = W^i - \int_0^\cdot \sigma_i\big(t,X_i(t)\big) \dt,
\end{align}
for all $i\in\I$, and this equation together with \eqref{eq:for-ter-BM} determines also the relations between the domestic forward Brownian motions; see \citet{Schloegl02} for the details (in particular Fig. 2). 
Therefore, the dynamics of the foreign LIBOR rate $L_i^f$ under the domestic forward measure $\Q_{i+1}$ are provided by
\begin{equation}
\begin{split}
\ud L^f_i(t) 
  = - L^f_i(t) \lambda_i\big(t,L^f_i(t)\big) \sigma_{i+1}\big(t,X_{i+1}(t)\big) \dt 
  + L^f_i(t) \lambda_i\big(t,L^f_i(t)\big) \ud W_t^{i+1}. 
\end{split}   
\end{equation}

\subsection{European quanto options and a local volatility model}

A quanto cap is a series of quanto caplets, where each quanto caplet is a call option on the foreign \lib rate struck at the domestic currency. 
In other words, a quanto caplet with strike price $K$ and expiry date $T_i$ pays at time $T_{i+1}$ the amount 
\begin{equation}\label{equ1}
\delta_i \big( L^f_i(T_i) - K \big)^+
\end{equation}
in units of domestic currency. 
Therefore, the price of a quanto caplet is provided by
\begin{equation}\label{equ6}
\mathbb{QC}(T_i,K) = \delta_i B(0,T_{i+1})  \,
	     \E_{i+1}\big[\big(L^f_i(T_i)-K\big)^+\big],
\end{equation}
where $\E_{i+1}$ denotes the expectation with respect to the domestic forward measure $\Q_{i+1}$.

In the sequel we will consider a local volatility model where each \lib rate and each FX forward rate are driven by `their own' one-dimensional, correlated 
Brownian motions, resulting in the following system of SDEs:
\begin{equation}\label{quantoeq}
\left\{
\begin{array}{rcl}
\ud L^f_i(t) 
  &=& -L^f_i(t) \lambda_i\big(t,L^f_i(t)\big) 
      \sigma_{i+1}\big(t,X_{i+1}(t)\big) \rho_i \dt \\
  &&\quad + L^f_i(t)\lambda_i\big(t,L^f_i(t)\big)\ud W^{L,i+1}_t \\
\ud X_{i+1}(t)
  &=& X(t, T_{i+1}) \sigma_{i+1}(t, X_{i+1}(t)) \ud W^{X,i+1}_t\\
\langle W^{L,i+1}, W^{X,i+1}\rangle &=& \rho_i, 
\end{array}
\right.
\end{equation}
with initial values $L_i^f(0),X_{i+1}(0)\in\R_+$, where $\lambda_i$ and $\sigma_i$ are $\R_+$-valued volatility functions and $\rho_i\in[-1,1]$.

Assuming that all the coefficients in \eqref{quantoeq} are deterministic, $L_i^f(t)$ is log-normally distributed and the price of a quanto caplet in \eqref{equ6} is given by a Black--Scholes type formula; see \textit{e.g.} \citet{BrigiMer06} or \citet{MusielaRut05}. 
In order to take into account the skew and smile effects observed in the fixed income and foreign exchange markets, we will consider a local volatility model and suppose that $\lambda_i(t, L^f_i(t))$ and $\sigma_{i+1}(t,X_{i+1}(t))$ are functions of $L^f_i(t)$ and $X_{i+1}(t)$ respectively. 
In that case, a closed form solution does not exist anymore and computing \eqref{equ6} numerically by Monte Carlo simulations or PDE methods is time consuming. 
Our objective therefore is to provide an approximation formula for \eqref{equ6} which is accurate enough and allows for an efficient \text{imp}lementation.

\begin{remark}
We assume throughout that the correlation between the forward \lib and the forward FX rate is deterministic and maturity-dependent. 
The extension to a time- and maturity-dependent correlation is straightforward.
\end{remark}

\section{An expansion approach to quanto pricing}
\label{sec:expa-quanto}

The main idea of expansion approaches to option pricing is to derive an asymptotic expansion of the option price in terms of quantities that are known and can be computed quickly, such as prices in a Black--Scholes model and Greeks.
This leads to a numerical scheme for the option price that is faster to compute than the corresponding PDE or Monte Carlo methods, while its accuracy can be \text{imp}roved by including additional terms in the expansion.
This section performs an analogous expansion for the price of a (generic) quanto option, and provides formulas for this option price in terms of a Black--Scholes model and Greeks, while the correlation between FX and LIBOR plays a crucial role.

In order to s\text{imp}lify notation, we will suppress the sub- and super-scripts related to the tenor and currency, and make use of the following notation: 
\begin{align}
L_t=L^f_i(t), \ X_t=X_{i+1}(t), \ W^{\cdot}=W^{\cdot,i+1}, \ \E=\E_{i+1}, \ \rho=\rho_i \ \text{ and } \ \ T=T_i .	
\end{align}
Considering the logarithms of the \lib and the FX rate, denoted by $Y=\ln L$ and $Z=\ln X$, the system of SDEs \eqref{quantoeq} takes the form
\begin{equation}\label{quantoeq3}
\left\{
\begin{array}{rclc}
\ud Y_t &=& \alpha(t, Y_t, Z_t)\dt + \lambda(t,Y_t)\ud W_t^L, 
  & Y_0 = \ln L_0 =y_0 \\
\ud Z_t &=& \beta(t, Z_t)\dt + \sigma(t, Z_t)\ud W_t^X,
  & Z_0 = \ln X_0 =z_0,\\ 	
\langle W^L, W^X \rangle &=& \rho, &
\end{array}
\right.
\end{equation}
where
\begin{equation}\label{quantoeq3-1}
\left\{
\begin{array}{rcl}
\alpha(t,y,z) &=& -[\frac{1}{2}\lambda^2(t,y) + \rho\lambda(t,y)\sigma(t,z)] \\
\lambda(t,y)  &=& \lambda_i(t,\e^y) \\
\beta(t,z)    &=& -\frac{1}{2}\sigma^2(t,z)\\
\sigma(t,z)   &=& \sigma_{i+1}(t,\e^z).
\end{array}
\right.
\end{equation}
Moreover, the price of a quanto option with generic payoff function $h:\R\to\R_+$ is provided by
\begin{equation}\label{Gprice}
\mathbb{QO}_h(T) = \delta B(0,T) \, \E[h(Y_T)]. 
\end{equation}
Taking $h(y)=(\e^y-K)^+$, we recover the quanto caplet in \eqref{equ1}--\eqref{equ6}.

\subsection{Closed-form formula under log-normal dynamics}

Assume that the local volatility coefficients are deterministic, time-dependent functions, in particular $\lambda(t,y)\equiv\lambda(t,y_0)$ and $\sigma(t,z)\equiv\sigma(t,z_0)$, for all $t\in[0,T]$. 
Then, $Y_T$ follows a normal distribution and we can derive a closed-form formula for the price of the quanto caplet defined in \eqref{equ1}. 

Indeed, using standard results from stochastic calculus, we have that
\begin{equation}
Y_T = y_0 -\frac12\Lambda(T) - \Sigma(T) + \int_0^T \lambda(t,y_0) \dW^L_t,
\end{equation}
where 
\begin{equation}
\Lambda(T) = \int_0^T \lambda^2(t,y_0) \dt
\ \ \text{ and } \ \
\Sigma(T) = \rho \int_0^T \lambda(t,y_0)\sigma(t,z_0) \dt.
\end{equation}
Thus, exactly as in \citet{BrigiMer06}, the price of the quanto caplet is given by the following, Black--Scholes type, formula:
\begin{align}\label{BSPriceFormula}
\mathbb{QC}^{BS}(T,K) 
 &= \delta B(0,T) \, \E\big[\big(\e^{Y_T}-K\big)^+\big] \nonumber \\
 &= \delta B(0,T) \, \mathbb C^{BS}(y_0)
\end{align} 
with
\begin{equation}
\mathbb C^{BS}(y_0) 
	= \e^{y_0-\Sigma(T)} \Phi(d_1) - \e^k\Phi(d_2),
\end{equation}
where $\Phi$ denotes the cdf of the standard normal distribution, $k = \ln K$, and
\begin{equation}
d_1 = \frac{y_0 - k - \Sigma(T) + \frac{1}{2}\Lambda(T)}{\sqrt{\Lambda(T)}}
\ \ \text{ and } \ \
d_2 = d_1 - \sqrt{\Lambda(T)}.
\end{equation}

\subsection{Expansion formulas under general dynamics}
\label{section:ApproximationFormula}

The aim of this subsection is to provide expansion formulas that approximate the price of a quanto option when we consider general local volatility dynamics 
for the LIBOR and the FX rate. 
Let us first introduce some notation, and some assumptions that allow us to derive these formulas.

\begin{ASRn}
The volatility functions $\lambda(\cdot,y)$ and $\sigma(\cdot,z)$ are of class $C^n$ in $y$ and $z$ respectively, for some $n\in\mathbb N$. 
In addition, these functions and their derivatives are uniformly bounded. 
\end{ASRn}

Let us introduce the following constants
\begin{align}
M_1^{\lambda} &:=  \max_{1 \leq i \leq n} \, \sup_{(t,y) \in [0,T] \times \R} 
|\partial^i_y \lambda(t,y)|, \ \ \ \
M_1^{\sigma} :=  \max_{1 \leq i \leq n} \, \sup_{(t,z) \in [0,T] \times \R} 
|\partial^i_z \sigma(t,z)|   \\
M_0^{\lambda} &:= \max\{ M_1^{\lambda}, \sup_{(t,y) \in [0,T] \times \R} 
|\lambda(t,y)| \}, \ \
M_0^{\sigma} := \max\{ M_1^{\sigma}, \sup_{(t,z) \in [0,T] \times \R} 
|\sigma(t,z)| \}   
\end{align}
and also
\begin{align}
M_1 &:= \max\{ M_1^{\lambda},M_1^{\sigma} \}, \ \ \ \
M_0 := \max\{ M_0^{\lambda},M_0^{\sigma} \}.   
\end{align}
Let us denote by $\mathcal{H}$ the space of functions with growth being at most exponential. 
In other words, a function $h$ belongs to $\mathcal{H}$ if $|h(x)| \leq c_1 \e^{c_2 |x|}$, for any $x$, for two positive constants $c_1$ and $c_2$.
Moreover, let $h^{(k)}$ denote the $k$-th derivative of the function $h$.

We will separate our analysis according to the smoothness of the payoff function, and distinguish between two cases: 
\begin{ASS1}
The payoff function $h$ belongs to $C_0^{\infty}(\R,\R)$, the space of real-valued infinitely differentiable functions with compact support. 
\end{ASS1}

\begin{ASS2}
The payoff function $h$ is almost everywhere differentiable. 
In addition $h$ and $h^{(1)}$ belong to $\mathcal{H}$. 
\end{ASS2}

\begin{remark}
The first assumption corresponds to (idealized) smooth payoff functions, while the second one corresponds to call and put options.
\end{remark}

In real markets, the correlation between the forward \lib rate and the forward FX rate is typically not very large. 
As an example, the empirical study in \citet{BoenSch03} found its value in the range $[-0.2, 0.2]$. 
Therefore, the following assumption is consistent with real market data.

\begin{ASR}
The correlation between the forward \lib rate and the forward FX rate is not perfect, \textit{i.e.}
\begin{equation}\label{hyporho}
|\rho| <1.
\end{equation}
\end{ASR}

In order to perform the infinitesimal analysis in the error estimates, we rely on smoothness properties which are not provided by the payoff functions, but 
rather by the law of the underlying stochastic models; this is related to Malliavin calculus. 
The following ellipticity assumption on the volatility of the forward \lib combined with Assumption $(\mathbb{R}_4)$---\textit{i.e.} Assumption $(\mathbb{R}_n)$ with $n=4$---guarantees that sufficient smoothness is available.

\begin{ASE}
The volatility of the forward \lib rate $\lambda$ does not vanish and for a positive constant $C_E$, one has
\begin{equation}
1 \leq \frac{\| \lambda\|_{\infty}}{\lambda_{\inf}} \leq C_E,
\end{equation}
where $\| g \|_{\infty} = \sup_{(t,y) \in [0,T] \times \R } | g(t,y) |$ and $\lambda_{\inf} = \inf_{(t,y) \in [0,T] \times \R } |\lambda(t,y)|$.
\end{ASE}


We consider now the following `proxy' or `Black--Scholes' processes:
\begin{eqnarray}\label{proxy0}
\left\{
\begin{array}{rclc}
\ud Y_t^0 &=& \alpha(t, y_0, z_0) \dt + \lambda(t, y_0) \dW_t^L, & Y_0^0=y_0, \\
\ud Z_t^0 &=& \beta(t,z_0) \dt + \sigma(t, z_0) \dW_t^X, & Z_0^0=z_0, \\
\langle W^L,W^X \rangle &=& \rho, &
\end{array}
\right.
\end{eqnarray}
and introduce a family of parametrized processes $(Y^\eta,Z^\eta)$, for $\eta\in[0,1]$, via the system of SDEs:
\begin{eqnarray}\label{pproxy1}
\left\{
\begin{array}{rclc}
\ud Y_t^\eta &=& \alpha\big( t, \eta Y^\eta_t+(1-\eta)y_0,\eta Z_t^\eta + (1-\eta)z_0\big)\dt &\\
	&& + \ \lambda\big(t, \eta Y_t^\eta+(1-\eta)y_0\big)\dW_t^L, 
	& Y_0^\eta=y_0,\\
\ud Z_t^\eta &=& \beta\big(t,\eta Z_t^\eta+(1-\eta)z_0\big)\dt & \\
	&& + \ \sigma \big(t, \eta Z_t^\eta+(1-\eta)z_0\big) \dW_t^X, & Z_0^\eta=z_0,\\
\langle W^L,W^X \rangle &=& \rho. &
\end{array}
\right.
\end{eqnarray}
Setting $\eta=1$, we recover the dynamics of the local volatility model in \eqref{quantoeq3} since $Y_t^{1}=Y_t$ and $Z_t^{1}=Z_t$, while for $\eta=0$ we 
recover the Black--Scholes proxy in \eqref{proxy0}. 

Assumption $(\mathbb R_4)$ yields that, almost surely for any $t\in[0,T]$, $\frac{\partial}{\partial \eta}(Y^\eta_t,Z^\eta_t)$ is $C^3$ with respect to $\eta$; see \textit{e.g.} \citet{Bell06} or \citet{Kun97}. 
Setting $Y^\eta_{i,t}=\frac{\partial^i Y_t^\eta}{\partial\eta^i},\ Z_{i,t}^\eta=\frac{\partial^i Z_t^\eta}{\partial\eta^i}$ and by a direct differentiation of the SDEs \eqref{pproxy1}, we get that 
\begin{eqnarray}
\left\{
\begin{array}{rcl}
\ud Z^\eta_{1,t} &=& (\eta Z^\eta_{1,t}+Z^\eta_t-z_0)[\beta_z\dt
  +\sigma_z \dW^X_{t}], \\
\ud Y_{1,t}^\eta &=& [(\eta Y^\eta_{1,t}+Y^\eta_t-y_0)\alpha_y+(\eta 
Z^\eta_{1,t}+Z_t^\eta-z_0)\alpha_z] \dt \\
 &&+ (\eta Y_{1,t}^\eta+Y^\eta_t-y_0)\lambda_y \dW^L_{t}, 
\end{array}
\right.
\end{eqnarray}
with $Y^\eta_{1,0} = Z^\eta_{1,0} = 0$, and
\begin{eqnarray}\label{SDESecondDerivative}
\left\{
\begin{array}{rcl}
\ud Z^\eta_{2,t} &=& (2Z^\eta_{1,t}+\eta Z^\eta_{2,t})
  [\beta_z \dt + \sigma_z \dW^X_t] \\
&&+ (\eta Z^\eta_{1,t}+Z^\eta_t-z_0)^2[\beta_{zz}\dt+\sigma_{zz}\dW^X_{t}], \\
\ud Y^\eta_{2,t} &=& (2Y^\eta_{1,t}+\eta Y^\eta_{2,t})[\alpha_y \dt + \lambda_y 
  \dW_t^L] + (2Z^\eta_{1,t}+\eta Z_{2,t}^\eta)\alpha_z\dt\\
&&+ 2(\eta Y^\eta_{1,t}+Y^\eta_{t}-y_0)(\eta 
  Z^\eta_{1,t}+Z_t^\eta-z_0)\alpha_{yz}\dt\\
&&+ [(\eta Y^\eta_{1,t}+Y^\eta_t-y_0)^2\alpha_{yy}+(\eta 
  Z^\eta_{1,t}+Z_t^\eta-z_0)^2\alpha_{zz}]\dt\\
&&+ (\eta Y^\eta_{1,t}+Y^\eta_t-y_0)^2\lambda_{yy}\dW^L_{t},
\end{array}
\right.
\end{eqnarray}
with $Y^\eta_{2,0} = Z^\eta_{2,0} = 0$, and
\begin{eqnarray}\label{SDEThirdDerivative}
\left\{
\begin{array}{rcl}
\ud Z^\eta_{3,t} &=& (3Z^\eta_{2,t}+\eta Z^\eta_{3,t})[\beta_z \dt + 
  \sigma_z \dW_z^X]\\
  &&+ 3( {2} Z^\eta_{1,t}+\eta Z^\eta_{2,t})(\eta Z^\eta_{1,t} + 
    Z^\eta_t-Z_0)[\beta_{zz}\dt+\sigma_{zz}\dW^X_t]\\
  &&+ (\eta Z_{1,t}^\eta+Z_t^\eta-Z_0)^3[\beta_{zzz}dt+\sigma_{zzz}\dW^X_t],
  \\
\ud Y_{3,t}^\eta &=& (3Y_{2,t}^\eta+\eta Y_{3,t}^\eta)[\alpha_y \dt + \lambda_y 
  \dW^L] + (3Z_{2,t}^\eta+\eta Z_{3,t}^\eta)\alpha_z \dt\\
 &&+ 3(\eta Y^\eta_{1,t}+Y_t^\eta-Y_0)(2Y_{1,t}^\eta+\eta 
     Y_{2,t}^\eta)[\alpha_{yy}\dt+\lambda_{yy}\dW_t^L)] \\
 &&+ 3(\eta Z_{1,t}^\eta+Z_t^\eta-Z_0)(2Z_{1,t}^\eta+\eta 
     Z_{2,t}^\eta)\alpha_{zz}\dt  \\
 &&+ 3[(\eta Y^\eta_{1,t}+Y_t^\eta-Y_0)(2Z_{1,t}^\eta+\eta Z_{2,t}^\eta) \\
 &&\quad + (\eta Z_{1,t}^\eta+Z_t^\eta-Z_0)(2Y_{1,t}^\eta+\eta 
	Y_{2,t}^\eta)]\alpha_{yz}\dt \\
 &&+  3[(\eta Y^\eta_{1,t}+Y_t^\eta-Y_0)^2(\eta Z_{1,t}^\eta + 
	Z_t^\eta-Z_0)\alpha_{yyz} \\ 
 &&\quad +  (\eta Y^\eta_{1,t}+Y_t^\eta-Y_0)(\eta Z_{1,t}^\eta + 
	Z_t^\eta-Z_0)^2\alpha_{zzy}]\dt \\
 &&+ (\eta Y^\eta_{1,t}+Y_t^\eta-Y_0)^3 \alpha_{yyy}\dt 
    + (\eta Z_{1,t}^\eta+Z_t^\eta-Z_0)^3\alpha_{zzz}\dt \\
 &&+ (\eta Y^\eta_{1,t}+Y_t^\eta-Y_0)^3 \lambda_{yyy} \dW^L_t,
\end{array} 
\right.
\end{eqnarray}
with $Y^\eta_{3,0} = Z^\eta_{3,0} = 0.$ 
Here, we have used the following shorthand notation for the first order derivatives of the coefficients of the SDEs
\begin{eqnarray}
\left\{
\begin{array}{rcl}
\alpha_x  &=& \frac{\partial \alpha}{\partial x}(t, y,z) \big|_{y=\eta Y_t^\eta+(1-\eta)y_0, z=\eta Z_t^\eta+(1-\eta)z_0}, \quad x\in\{y,z\}, \\
\lambda_y &=& \frac{\partial \lambda}{\partial y}(t,y) \big|_{y=\eta Y^\eta_t+(1-\eta)y_0}, \\
\beta_z   &=& \frac{\partial \beta}{\partial z}(t,z) \big|_{z=\eta Z^\eta_t+(1-\eta)z_0},  \\ 
\sigma_z  &=& \frac{\partial \sigma}{\partial z}(t,z) \big|_{z=\eta Z^\eta_t+(1-\eta)z_0},
\end{array}
\right. 
\end{eqnarray}
and analogously for higher order derivatives.


Let us now introduce the main tools of this method, which are expansions of the random variable and the payoff function of the quanto option around known values.
In order to keep the notation simple, we set $Y_{i,t}=\frac{\partial^i Y^\eta_t}{\partial \eta^i}|_{\eta=0}$, $Z_{i,t}=\frac{\partial^i Z^\eta_{t}}{\partial \eta^i}|_{\eta =0}$.
Then, by performing a Taylor expansion of $Y_T$ around zero, we get that
\begin{equation}
Y_T = Y_T^0 + Y_{1,T} +  \frac{1}{2} Y_{2,T}
  + \frac{1}{2} \int_{0}^1 Y^\eta_{3,T}(1-\eta)^2 \ud\eta.
\end{equation}
The dynamics of the proxy model in \eqref{proxy0} yield that $Y_T^0$ is a Gaussian random variable with mean $m_T^0$ and variance $\sqrt{V_T^0}$, where
\begin{equation}
m_T^0 = y_0 + \int_0^{T} \alpha(t, y_0, z_0) \dt
\quad\text{ and }\quad 
V_T^0 = \int_0^{T} \lambda^2(t, y_0) \dt.
\end{equation}
Performing now a Taylor expansion of the payoff in \eqref{Gprice} around $h(Y_T^0)$, we arrive at a formula of the following form:
\begin{equation}\label{eq:main-idea}
\E[h(Y_T)] = \E[h(Y^0_{T})] + \mbox{Corrections terms} + \mbox{Error}.
\end{equation}
The first term $\E[h(Y^0_{T})]$ constitutes the leading order contribution, it is explicitly known (via an analytical formula analogous to \eqref{BSPriceFormula} for the payoff function $h$), but as an approximation alone is not accurate enough. 
Therefore, in the sequel we will derive correction terms in order to achieve better accuracy. 
These correction terms are represented as a combination of Greeks of the option price formula \eqref{BSPriceFormula}. 
Hence, the numerical evaluation of all these terms is straightforward, with a computational cost equivalent to the analytical formula \eqref{BSPriceFormula}.

\subsection{Definitions and notation}
 
Before providing the main results, let us introduce some definitions and notation that will be used in the sequel.

\begin{definition}[Integral Operator] 
The integral operator $\omega$ is defined as follows: for any integrable function $l$, set
\begin{equation}
\omega(l)_t^{T} = \int_t^{T}l_u\ud u
\end{equation}
for $t\in[0,T]$. 
Similarly, for integrable functions $(l_1,l_2)$ and $t\in[0,T]$ set
\begin{equation}
\omega(l_1,l_2)^{T}_t 
  = \omega(l_1\omega(l_2)_\cdot^T)^{T}_t
  =\int_{t}^T l_{1,r}\Big( \int_r^{T}l_{2,s}\ds \Big) \ud r.
\end{equation}
This can be easily iterated to define $\omega(l_1,l_2,\cdots{,l_n})_t^{T} = \omega (l_1\omega (l_2,\cdots,l_n)_\cdot^T)_t^{T}$.
\end{definition}

\begin{definition}[Greeks]\label{def:greeks}
Let $h$ be an appropriate payoff function (such that the expression below makes sense). 
We set, for $i\geq{0},$
\begin{equation}
g_i^{h}(Y^0_{T}) 
  = \frac{\partial^i }{\partial \epsilon^i} \E\big[h\big(Y^0_T+\epsilon\big)\big] \big|_{\epsilon=0}.
\end{equation}
\end{definition}
 
\begin{remark}[Generic constants] 
We use the notation $A \leq_c B$ to assert that $A \leq cB$, where $c$ is a positive constant depending on the model parameters, on $M_0$, $M_1$, $T$, $C_E$ and on other universal constants. 
The constant $c$ may change from line to line, but remains bounded when the model parameters go to $0$. 
\end{remark}

\begin{remark}[Notation for the coefficients]\label{rem:notation-coeffs}
The coefficients $\alpha, \beta, \lambda, \sigma$ and their derivatives will be evaluated from now on at the initial values $(y_0,z_0)$, \textit{i.e.} when we write $\alpha, \beta, \lambda, \sigma$ we mean $\alpha(\cdot,y_0,z_0), \beta(\cdot,z_0), \lambda(\cdot,y_0), \sigma(\cdot,z_0)$, and the same holds for their derivatives. 
We will sometimes also use the subscript $t$ when we want to stress their dependence on time.
\end{remark}

\subsection{Main results}

We are now ready to state the main results of this work, that provide second and third order expansions of an option price around the proxy model, thus making 
precise the formula in \eqref{eq:main-idea}. 
The proofs are deferred to Section \ref{section:Proofs}.

\begin{theorem}[2nd order expansion in price]
\label{ApproxFormulaProbaApproa2order} 
Assume that conditions $(\mathbb R_3)$, $(\mathbb S_2)$, $(\mathbb{ELL})$ and $(\mathbb{RHO})$ are in force. 
Then, the second order expansion of the option price takes the form:
\begin{eqnarray}\label{2ndOrderExpansionFormula}
\begin{array}{rcl}
\E[h(Y_T)] 
&=& \E[h(Y^0_T)] + \omega(\lambda^2,\lambda_y\lambda)_{0}^{T} \bigg[ 
  \frac{1}{2}g_1^h(Y^0_T) - \frac{3}{2}g_2^h(Y^0_T) + g_3^h(Y^0_T) \bigg] 
  \\[1.em]
&& + \rho \bigg[ g_1^h(Y^0_T) [ \omega(\lambda\sigma, \lambda_y\lambda)_{0}^{T} 
  + \frac{1}{2} \big( \omega(\lambda^2 , \lambda_y\sigma)_{0}^{T} 
  + \omega(\sigma^2 , \lambda \sigma_z )_{0}^{T} \big)  ] \\[1.em]
&&\qquad - g_2^h(Y^0_T)[ \omega( \lambda \sigma, \lambda_y\lambda )_{0}^{T} +
  \omega( \lambda^2 , \lambda_y\sigma)_{0}^{T}] \bigg]  \\[1.em]
&& + \rho^2 \bigg[ g_1^h(Y^0_T) \omega( \lambda \sigma,\lambda_y\sigma )_{0}^{T} 
   - g_2^h(Y^0_T) \omega( \lambda\sigma,\lambda \sigma_z )_{0}^{T} \bigg] 
+ \mbox{\rm Error}_2.
\end{array}
\end{eqnarray}
Additionally, the error estimate is provided by
\begin{equation}
|\mbox{\rm Error}_2| 
\leq_c \bigg[ \| h^{(1)}(Y^0_T) \|_2 + \int_0^1 \| 
  h^{(1)} (\eta Y_T + (1-\eta)Y_T^0) \|_2 \ud\eta \bigg]
  \frac{M_0}{\lambda_{\mathrm{inf}} (1-\rho^2)}  M_1 M_0^2 T^{\frac{3}{2}}.
\end{equation}
\end{theorem}

\begin{theorem}[3rd order expansion in price]
\label{ApproxFormulaProbaApproa3order}
Assume that conditions $(\mathbb R_4)$, $(\mathbb S_2)$, $(\mathbb{ELL})$ and $(\mathbb{RHO})$ are in force. 
Then, the third order expansion of the option price takes the form:
\begin{equation}\label{3rdOrderExpansionFormula}
\E[h(Y_T)] 
= \E[h(Y^0_T)] + \sum_{j=1}^6 \gamma_{0,j,T} g_j^h(Y^0_T) 
+ \sum_{i=1}^4 \gamma_{i,T}\rho^i + \mbox{\em Error}_3,
\end{equation}
where
\begin{eqnarray}
\left\{
\begin{array}{lll}
\gamma_{0,1,T} &=& \frac{1}{2} (A_{1,T} -  A_{2,T} - A_{3,T} ) 
  - \frac{1}{2} B_{1,T} - \frac{1}{4} ( B_{2,T} + B_{3,T} ) \\ 
\gamma_{0,2,T} &=& -\frac{3}{2} A_{1,T} + \frac{1}{2} (A_{2,T} + A_{3,T}) 
 + \frac{7}{2}B_{1,T} + \frac{5}{4}(B_{3,T}+B_{2,T}) 
 +  \frac{1}{2} C_{33,T} + \frac{1}{4}C_{32,T} \\
\gamma_{0,3,T} &=& A_{1,T} - 6B_{1,T} - 2(B_{3,T}+B_{2,T}) 
  - \frac{3}{2}C_{32,T} -3C_{33,T} \\
\gamma_{0,4,T} &=& 3B_{1,T} + B_{2,T} +B_{3,T} + \frac{13}{4}C_{32,T} 
  +  \frac{13}{2}C_{33,T} \\
\gamma_{0,5,T} &=& -3 C_{32,T} -6C_{33,T} \\
\gamma_{0,6,T} &=& C_{32,T} + 2C_{33,T} 
\end{array}
\right.
\end{eqnarray}
with 
\begin{eqnarray}
\left\{
\begin{array}{l}
A_{1,T} = \omega(\lambda^2,\lambda\lambda_y)_0^T  \quad\quad\quad
A_{2,T} = \omega(\lambda^2,\lambda\lambda_{yy})_0^T \quad\quad
A_{3,T} = \omega(\lambda^2,(\lambda_y)^2)_0^T   \\
B_{1,T} = \omega(\lambda^2,\lambda\lambda_y,\lambda\lambda_y)_0^T \quad
B_{2,T} = \omega(\lambda^2,\lambda^2,\lambda\lambda_{yy})_0^T \ \
B_{3,T} = \omega(\lambda^2,\lambda^2,(\lambda_y)^2)_0^T \\ 
C_{32,T} = \omega(\lambda^2,\lambda\lambda_y,\lambda^2,\lambda\lambda_y)_0^T
\quad
C_{33,T} = \omega(\lambda^2,\lambda^2,\lambda\lambda_y,\lambda\lambda_y)_0^T.
\end{array}
\right.
\end{eqnarray}
The expressions for the coefficients $\gamma_{i,T}$, $i=1,2,3,4$, are provided respectively by \eqref{gamma1}, \eqref{gamma2}, \eqref{gamma3} and \eqref{gamma4} in Appendix \ref{DerivThirdOrderFormula}. 
Additionally, the error estimate is given by 
\begin{equation}
| \mbox{\em Error}_3 | 
  \leq_c \bigg[ \| h^{(1)}(Y^0_T) \|_2 + \int_0^1 
    \| h^{(1)}(\eta Y_T + (1-\eta)Y_T^0) \|_2 \ud \eta \bigg] 
    \frac{M_0^5 M_1}{(\lambda_{\mathrm{inf}}(1-\rho^2))^2} T^2. 
\end{equation}
\end{theorem}

\begin{remark}[Sanity check] 
Let $\rho=0$, then the quanto (drift) adjustment in the \lib SDE \eqref{quantoeq3} vanishes, and we recover the second and third order approximation formulas given respectively by Theorems 2.1 and 2.3 in \citet{BenGobetMiri10a}.  

If $\rho \neq 0$, the second and third order expansion formulas \eqref{2ndOrderExpansionFormula} and \eqref{3rdOrderExpansionFormula} provide some information about the \text{imp}act of the quanto (drift) adjustment in the option prices, in terms of a polynomial function of the correlation $\rho$.  
\end{remark}

Consider a call option with payoff $h(y) = (\e^y - \e^k)_+$, then the theorems above provide an approximation formula for its price in log variables. 
In that case, $\E[h(Y^0_T)] = \mathbb C^{BS}(y_0)$ corresponds to the Black--Scholes price given by \eqref{BSPriceFormula}. 
In order to compute the correction terms, we need to calculate the derivatives of $\mathbb C^{BS}(y_0)$ w.r.t. $y_0$. 
Below is a useful lemma allowing to calculate them in a systematic way using Hermitte polynomials.

\begin{lemma} \label{deryBSlemma}
Let $n \geq 1$, then we have
\begin{equation} \label{deryBS}
\frac{\partial^n \mathbb C^{BS}(y_0)}{\partial y_0^n} 
  = \e^{y_0- \Sigma(T)} \left[\Phi(d_1) + 
  1_{\{n \geq 2\}} \Phi^{'}(d_1) \sum_{j=1}^{n-1} \binom{n-1}{j} 
  (-1)^{j-1} \frac{H_{j-1}(d_1)}{(\bar{\lambda}\sqrt{T})^j} \right],
\end{equation}
where $H_j$, $j\in\N$, denotes the Hermitte polynomials defined as 
\begin{equation}
H_j(x) = (-1)^j \e^{\frac{x^2}{2}} \partial^n_{x^n}(\e^{-\frac{x^2}{2}}), 
\  \ j \in \N .	
\end{equation} 
\end{lemma}
The proof is provided in Appendix \ref{app:proof}.

\section{Analysis and proofs} 
\label{section:Proofs}

This section provides the derivation of the expansion formulas for quanto pricing presented in Theorems \ref{ApproxFormulaProbaApproa2order} and \ref{ApproxFormulaProbaApproa3order}, as well as an analysis of the corresponding error terms.
After some preliminary results, the expansion formulas and the corresponding error estimates for the second and third order expansions are presented in Section \ref{sec:pe-ee}.
The derivation of the Greeks for the second order expansion is presented in Section \ref{subsection:CompGreeks}, while the details for the Greeks of the third order expansion are deferred to the appendix for the sake of brevity.

\subsection{Auxiliary results}

We start with some results that are useful for the subsequent error analysis. 
The $L^p$-estimates follow from the work of \citet[Theorem 5.1]{BenGobetMiri10a}, thus their proof is omitted.
As usual, the $L^p$-norm of a real random variable $Z$ is provided by $\|Z\|_p = \big(\E[|Z|^p]\big)^{\frac1p}, \, p\geq 1$.

\begin{lemma}[$L^p$-estimates]
\label{LpEstimates} 
Assume that condition $(\mathbb R_4)$ is in force. 
Then, for all $p\geq1$ and $i=1,2,3$, we have
\begin{equation}\label{Z0Y0}
\sup_{t \in [0,T], \eta \in [0,1]} \|  Z_t^\eta-z_0 \|_p \leq_c  M_0\sqrt{T}, \, 
\sup_{t \in [0,T], \eta \in [0,1]} \|  Y_t^\eta-y_0 \|_p \leq_c  M_0\sqrt{T},
\end{equation}
\begin{equation}\label{ZiYi}
\sup_{t \in [0,T], \eta \in [0,1]} \|  Z_{i,t}^\eta \|_p \leq_c  
  M_1M_0^iT^{\frac{i+1}{2}}, \,	
\sup_{t \in [0,T], \eta \in [0,1]} \|  Y_{i,t}^\eta \|_p \leq_c  
M_1M_0^iT^{\frac{i+1}{2}}.
\end{equation}		
\end{lemma}

The following lemmata are used repeatedly in order to derive the analytical formulas in Theorems \ref{ApproxFormulaProbaApproa2order} and \ref{ApproxFormulaProbaApproa3order}. 
An application of It\^o's lemma to $(\int_t^Tf_s\ds)Z_t$ yields the following result. 

\begin{lemma}\label{IBP1}
Let $f$ be a continuous (or piecewise continuous) function and $Z$ be a continuous semimartingale with $Z_0=0$. 
Then
\begin{equation}
\int_{0}^{T} f_{t} Z_{t} \dt
  = \int_{0}^{T} \Big( \int_{t}^{T} f_{s} \ds \Big) \ud Z_{t}
  =\int_0^T \omega(f)_t^T \ud Z_t.
\end{equation}
\end{lemma}

The lemma below follows directly from the duality relationship in the Malliavin calculus (see \textit{e.g.} \citet[Lemma 1.2.1,~p.25]{Nualart05}) and by identifying It\^o's integral and the Skorohod operator for adapted integrands.

\begin{lemma}\label{IBP2} 
Let $u$ be a square integrable, progressively measurable process and assume $h$ satisfies $(\mathbb{S}_1)$. 
Then, for any $i \geq 0$, it holds:
\begin{multline}
\E\left[\left( \int_{0}^{T} u_{t}\dW^{\alpha}_{t} \right) 
  h^{(i)}\left(\int_{0}^{T} \lambda(t,y_0) \dW^L_{t} \right)\right] \\
 = \E\left[\left( \int_{0}^{T} u_{t} \lambda(t,y_0) \ud\langle W^{\alpha}, W^L \rangle_t \right) 
 	h^{(i+1)}\left(\int_{0}^{T} \lambda(t,y_0) \dW^L_{t} \right)\right]
\end{multline}
with $\alpha\in\{L,X\}$ and $h^{(i)}(x) = \frac{\ud^i}{\dx^i}h(x), \ i \in \N$. 
Moreover, if $u$ and $\langle W^{\alpha},W^L \rangle$ are deterministic, then
\begin{equation}
\E\bigg[\bigg( \int_{0}^{T} u_{t} \dW^{\alpha}_{t}\bigg) 
  h^{(i)}\bigg(\int_{0}^{T} \lambda(t,y_0) \dW^L_{t} \bigg)\bigg] 
 = \int_{0}^{T} u_{t} \lambda(t,y_0) \ud \langle W^{\alpha}, W^L \rangle_t g_{i+1}^{\tilde{h}}(Y_T^0),
\end{equation}
where $\tilde{h}(x) = h(x-m_T^0)$.
\end{lemma}

\subsection{Price expansions and error estimates}
\label{sec:pe-ee}

We are now ready to provide the details in the derivation of the expansion formulas and the corresponding error estimates.
We start with the analysis of the second order approximation, and divide the proof of Theorem \ref{ApproxFormulaProbaApproa2order} in several steps. 
First, we assume that the payoff $h$ is smooth and establish error estimates that depend only on $h^{(1)}$, the first derivative of $h$. 
To this end, we use Malliavin calculus and provide tight estimates on the Malliavin derivatives of the parametrized process. 
Then, we can approximate $h$ under $(\mathbb S_2)$ by a sequence of smooth payoffs using a density argument.
This last step is standard by now, hence we omit it for the sake of brevity.

\subsubsection{Second order error analysis}

As outlined in the previous section, we perform first a Taylor expansion of $Y_T$ around $Y^0_T$, that yields
\begin{equation}\label{SdOrderExpYT}
Y_T = Y_T^0 + Y_{1,T} + \int_0^1 Y^\eta_{2,T}(1-\eta) \ud\eta,  
\end{equation}
then another Taylor expansion for the smooth payoff $h$, and then take expectations. 
Thus we obtain
\begin{align}\label{SdOrderExpPayoff}
\E[h(Y_T)] &= \E[h(Y^0_T)] +  \E[h^{(1)}(Y^0_T)(Y_T - Y^0_T)] \nonumber \\
  &\quad+ \E\Big[ (Y_T - Y^0_T)^2 \int_0^1 h^{(2)}(\eta Y_T + (1-\eta)Y^0_T) 
      (1-\eta)\ud\eta \Big].
\end{align}
Using \eqref{SdOrderExpYT}, \eqref{SdOrderExpPayoff} can be written as
\begin{equation}\label{SdOrderExpPayoffanderror}
\E[h(Y_T)] = \E[h(Y^0_T)] +  \E[h^{(1)}(Y^0_T)Y_{1,T}] + \mbox{Error}_2.
\end{equation}
where 
\begin{equation}\label{error2}
\mbox{Error}_2 
  = \E\big[ h^{(1)}(Y^0_T)R_T^{1,Y} \big] 
    + \E\Big[ (R_T^{0,Y})^2 \int_0^1 h^{(2)}(\eta Y_T + (1-\eta)Y^0_T) 
      (1-\eta) \ud\eta \Big]
\end{equation}
with 
\begin{equation}\label{res1}
R_T^{0,Y} = \int_0^1 Y^\eta_{1,T}\ud\eta
\quad \text{ and } \quad
R_T^{1,Y} = \int_0^1 Y^\eta_{2,T}(1-\eta)\ud\eta.
\end{equation}
Using \eqref{ZiYi} with $i=2$ in Lemma \ref{LpEstimates} and the Cauchy--Schwarz inequality, the first term in \eqref{error2} is estimated as
\begin{equation}
\big| \E\big[ h^{(1)}(Y^0_T)R_T^{1,Y} \big] \big| 
  \leq_c \| h^{(1)}(Y^0_T) \|_2 M_1M_0^2 T^{\frac{3}{2}}.
\end{equation}

The second term in \eqref{error2} requires some additional work because of $h^{(2)}$. 
We use the integration-by-parts formula in the Malliavin calculus to write it using $h^{(1)}$ only. 
For this, we rely on Lemma \ref{ErrorAnalysisIBP} and refer to Appendix \ref{sec:appendix} for notation related to the Malliavin calculus. 
Let us apply this result to $V = (R_T^{0,Y})^2$, such that we can write 
\begin{multline}
\E\Big[ \big(R_T^{0,Y}\big)^2 \int_0^1 
  h^{(2)}\big(\eta Y_T + (1-\eta)Y^0_T\big)(1-\eta)\ud\eta \Big]  \\
= \int_0^1 \E\Big[ h^{(2)}\big(\eta Y_T + (1-\eta)Y^0_T\big) 
    \big(R_T^{0,Y}\big)^2 \Big] (1-\eta) \ud\eta  \\
= \int_0^1 \E\Big[ h^{(1)}\big(\eta Y_T + (1-\eta)Y^0_T\big) V^\eta_{1} \Big] (1-\eta) \ud\eta.
\end{multline}
Using now the $L^p$ estimates in Lemmata \ref{LpEstimates} and \ref{EstimatesMalliavinDerivatives}, we can show easily that
\begin{equation}
\| (R_T^{0,Y})^2  \|_{1, 2p} \leq_c (M_1M_0T)^2
\end{equation}
and get
\begin{equation}
\| V^\eta_{1}  \|_p \leq_c 
\frac{(M_1\sqrt{T})^2(M_0\sqrt{T})^2}{(1-\rho^2) \lambda_{\mathrm{inf}} 
\sqrt{T}}.
\end{equation}
Therefore, we can deduce that
\begin{multline}
| \E\Big[ \big(R_T^{0,Y}\big)^2 \int_0^1 
  h^{(2)}\big(\eta Y_T + (1-\eta)Y^0_T\big) (1-\eta)\ud\eta \Big] | \\
\leq_c \int_0^1 \| h^{(1)}\big(\eta Y_T + (1-\eta)Y_T^0)\big \|_2 \ud \eta 
  \frac{M_0}{\lambda_{\mathrm{inf}} (1-\rho^2)}  M_1M_0^2 T^{\frac{3}{2}}.
\end{multline}
Because $\lambda_{\mathrm{inf}} \leq_c M_0$ and $\frac{1}{(1-\rho^2)} \geq 1$, we finally obtain
\begin{equation}
| \mbox{Error}_2 | 
  \leq_c \bigg[ \| h^{(1)}(Y^0_T) \|_2 + \int_0^1 \| 
    h^{(1)}\big(\eta Y_T + (1-\eta)Y_T^0\big) \|_2 \ud\eta \bigg] 
  \frac{M_0^3M_1}{\lambda_{\mathrm{inf}} (1-\rho^2)} T^{\frac{3}{2}}. 
\end{equation}

Thus far, we have bounded the error using only $h^{(1)}$ for a smooth function $h$. 
In order to obtain a similar error bound under the assumption that $h$ satisfies ($\mathbb S_2$), we can use a density or regularization argument to approximate $h$ by a sequence of smooth functions as in \citet[Section 5.2, Step 4]{MiriGobetBen02}.

\subsubsection{Third order error analysis} 

We follow again the same strategy as for the second order case. 
By a Taylor expansion of $Y_T$ around $Y^0_T$, we have
\begin{equation}\label{ThirdOrderExpYT}
Y_T = Y_T^0 + Y_{1,T} + \frac{1}{2}Y_{2,T} 
  + \int_0^1 Y^\eta_{3,T}\frac{(1-\eta)^2}{2} \ud\eta,
\end{equation}
and by performing again a Taylor expansion for a smooth payoff $h$ and taking expectations we obtain	
\begin{align}\label{ThirdOrderExpPayoff}
\nonumber 
\E[h(Y_T)] &= \E[h(Y^0_T)] +  \E\big[h^{(1)}(Y^0_T)(Y_T - Y^0_T)\big] 
  + \E\left[\frac{1}{2} h^{(2)}(Y^0_T) (Y_T - Y^0_T)^2 \right] \\
&\quad + \E\left[(Y_T - Y^0_T)^3 \int_0^1 h^{(3)}(\eta Y_T + 
  (1-\eta)Y^0_T) \frac{(1-\eta)^2}{2} \ud\eta \right].
\end{align}
Using \eqref{ThirdOrderExpYT}, the latter becomes
\begin{align}\label{ThirdOrderExpPayoffanderror}
\E \left[h(Y_T) \right] 
 &= \E\left[h(Y^0_T)\right] +  \E\left[h^{(1)}(Y^0_T)Y_{1,T}\right] 
  + \E\left[h^{(1)}(Y^0_T)\frac{Y_{2,T}}{2}\right] \nonumber\\
 &\quad + \E\left[\frac{h^{(2)}(Y^0_T)}{2}Y^2_{1,T}\right] + \mbox{Error}_3,
\end{align}
where
\begin{align}
\mbox{Error}_3 
  &= \E\left[ h^{(1)}(Y^0_T) \int_0^1 Y^\eta_{3,T}\frac{(1-\eta)^2}{2} 
    \ud\eta \right] \label{FirstTermError3}\\
  &\quad + \E\left[ (Y_T - Y^0_T)^3 \int_0^1 h^{(3)}(\eta Y_T + (1-\eta)Y^0_T) 
      \frac{(1-\eta)^2}{2} \ud\eta \right]  \label{SecondTermError3}\\
  &\quad + \E\left[ \frac{h^{(2)}(Y^0_T)}{2} \left[ (Y_T - Y_T^0)^2  -Y_{1,T}^2 
      \right]\right]. \label{ThirdTermError3}
\end{align}

\noindent Let us bound each term in the error separately.
The first term \eqref{FirstTermError3}, using \eqref{ZiYi} with $i=2$ in Lemma \ref{LpEstimates} and the Cauchy--Schwarz inequality, is estimated by
\begin{equation}\label{FirstTermError3Estimate}
\bigg| \E\left[ h^{(1)}(Y^0_T) \int_0^1 Y^\eta_{3,T} \frac{(1-\eta)^2}{2} 
  \ud\eta\right] \bigg| 
 \leq_c \| h^{(1)}(Y_T^0) \|_2 M_1M_0^3 T^2.
\end{equation}

The second term \eqref{SecondTermError3} is handled as in the previous section. 
	We recall that $Y_T-Y^0_T=R_T^{0,Y}=\int_0^1 Y^\eta_{1,T}\ud\eta$ and apply Lemma \ref{ErrorAnalysisIBP} with $k=2$ to $V = (R_T^{0,Y})^3$ such that we can write 
\begin{multline}
\E\left[ (R_T^{0,Y})^3 \int_0^1 h^{(3)} \big( \eta Y_T + (1-\eta)Y^0_T \big) 
	\frac{(1-\eta)^2}{2} \ud\eta \right] \\
 = \int_0^1 \E\left[ h^{(3)} \big(\eta Y_T + (1-\eta)Y^0_T\big)			
	(R_T^{0,Y})^3 \right] \frac{(1-\eta)^2}{2} \ud\eta  \\
 = \int_0^1 \E\left[ h^{(1)}\big(\eta Y_T + (1-\eta)Y^0_T\big) V_2^\eta \right] \frac{(1-\eta)^2}{2}\ud\eta.
\end{multline}    
Using the $L^p$ estimates in Lemmata \ref{LpEstimates} and \ref{EstimatesMalliavinDerivatives}, we show easily that
\begin{equation}
\| (R_T^{0,Y})^2  \|_{2, 2p} \leq_c (M_1M_0T)^3,
\end{equation}
hence
\begin{equation}
\| V^\eta_2 \|_p \leq_c 
\bigg( \frac{M_0}{\lambda_{\mathrm{inf}}(1-\rho^2)}\bigg)^2 M_1 M_0^3 T^2.
\end{equation}
Therefore, we can deduce that
\begin{multline}\label{SecondTermError3Estimate}
\qquad \bigg| \E\bigg[ (R_T^{0,Y})^3 \int_0^1 h^{(3)}\big(\eta Y_T + 
  (1-\eta)Y^0_T\big) \frac{(1-\eta)^2}{2} \ud\eta\bigg] \bigg| \\
 \leq_c \int_0^1 \| h^{(1)}\big(\eta Y_T + (1-\eta)Y_T^0\big) \|_2 \ud\eta 
   \frac{M_0^5M_1}{(\lambda_{\mathrm{inf}}(1-\rho^2))^2} T^2.
\end{multline}

As for the third term \eqref{ThirdTermError3}, let us first provide a more explicit representation of $(Y_T - Y_T^0)^2 - Y_{1,T}^2$. 
We define
\begin{equation}\label{MoreGeneralErrorStudy}
f(\eta) = \big( Y^\eta_T- Y^0_T\big)^2
\end{equation}
and perform a second order Taylor expansion around $0$ to get
\begin{equation}\label{TaylorErrorAnalysis}
f(\eta) =  f(0) + f^{(1)}(0)\eta + f^{(2)}(0)\frac{\eta^2}{2} + 
\int_0^\eta \frac{(\eta-t)^2}{2}f^{(3)}(t) \dt  
\end{equation}
where
\begin{eqnarray}
\left\{
\begin{array}{rclrcl}
f(0) &=& f^{(1)}(0) = 0 &\quad
  f^{(1)}(\eta) &=& 2 Y_{1,T}^\eta(Y_T^\eta-Y_T^0)\\
f^{(2)}(0) &=& 2 Y_{1,T}^2 &
  f^{(2)}(\eta) &=& 2 \big[ Y_{2,T}^\eta(Y_T^\eta-Y_T^0) + (Y_{1,T}^\eta)^2 
  \big]\\
f^{(3)}(0) &=& 6 Y_{1,T} Y_{2,T} &
  f^{(3)}(\eta) &=& 2 \big[ Y_{3,T}^\eta(Y_T^\eta-Y_T^0) + 3Y_{1,T}^\eta 
  Y_{2,T}^\eta \big] .
\end{array}
\right.
\end{eqnarray}
Setting $\eta=1$ in \eqref{TaylorErrorAnalysis}, we obtain
\begin{equation}\label{Power2ErrorRepresentation}
(Y_T - Y_T^0)^2  = Y_{1,T}^2 + \int_0^1 (1-\eta) \big[ 
  Y_{3,T}^\eta(Y_T^\eta-Y_T^0) + 3Y_{1,T}^\eta Y_{2,T}^\eta \big] \ud\eta.
\end{equation}
Replacing \eqref{Power2ErrorRepresentation} into \eqref{ThirdTermError3} and using Fubini's theorem, we get
\begin{multline}
\qquad \E \left[ \frac{h^{(2)}(Y^0_T)}{2} \int_0^1 (1-\eta) \big[ 
  Y_{3,T}^\eta(Y_T^\eta-Y_T^0) + 3Y_{1,T}^\eta Y_{2,T}^\eta \big] \ud\eta 
  \right] \\
= \int_0^1 \frac{(1-\eta)}{2} \E \bigg[ h^{(2)}(Y^0_T) 
  (Y_{3,T}^\eta(Y_T^\eta-Y_T^0) + 3Y_{1,T}^\eta Y_{2,T}^\eta ) \ud\eta \bigg] \\
= \int_0^1 \frac{(1-\eta)}{2} \E \bigg[ h^{(1)}(Y^0_T) V^\eta_1 d\eta \bigg],
\end{multline}
where for the last equality we have applied the integration-by-parts formula of Lemma \ref{ErrorAnalysisIBP} with $V = Y_{3,T}^\eta(Y_T^\eta-Y_T^0) + 
3Y_{1,T}^\eta Y_{2,T}^\eta$ for $k=1$. 
Applying now the Cauchy--Schwartz inequality, we get the following error estimate
\begin{equation}\label{Error3rdtermestimates}
\bigg| \E\left[ \frac{h^{(2)}(Y^0_T)}{2} \left[ (Y_T - Y_T^0)^2  -Y_{1,T}^2 
  \right] \right] \bigg| 
\leq_c \|h^{(1)}(Y_T^0) \|_2 \int_0^1 \| V^\eta_1 \|_2 \ud\eta,  
\end{equation}
while the $L^p$ estimates in Lemmata \ref{LpEstimates} and \ref{EstimatesMalliavinDerivatives} yield, for $p\geq 1$, that
\begin{equation}
\| V \|_{1,2p} \leq_c M_1^2 M_0^3 T^{\frac{5}{2}}
\end{equation}
and
\begin{equation}
\| V^\eta_1 \|_p \leq_c \bigg( \frac{M_0}{\lambda_{\mathrm{inf}} (1-\rho^2)} \bigg) M_1M_0^3 T^2.
\end{equation}
Therefore, the third error term \eqref{Error3rdtermestimates} is estimated by
\begin{multline}\label{Error3rdtermestimatesFinal}
\bigg| \E\left[ \frac{h^{(2)}(Y^0_T)}{2} \left[ (Y_T - Y_T^0)^2 - Y_{1,T}^2 
  \right]\right] \bigg| 
\leq_c \|h^{(1)}(Y_T^0) \|_2 \bigg( \frac{M_0}{\lambda_{\mathrm{inf}} 
  (1-\rho^2)} \bigg) M_1M_0^3 T^2.  
\end{multline}

Finally, using again that $\lambda_{\mathrm{inf}} \leq_c M_0$ and $1 \le \frac{1}{(1-\rho^2)}$, and by regrouping all the estimates in
\eqref{FirstTermError3Estimate}, \eqref{SecondTermError3Estimate} and \eqref{Error3rdtermestimatesFinal}, the third order error can be estimated as follows:
\begin{multline}	
| \mbox{Error}_3 | 
\leq_c \bigg[ \| h^{(1)}(Y^0_T) \|_2 + \int_0^1 \| h^{(1)}(\eta Y_T + (1-\eta)Y_T^0) \|_2 \ud\eta  \bigg] 
\\ \times \bigg( 
  \frac{M_0}{\lambda_{\mathrm{inf}}(1-\rho^2)}\bigg)^2 M_1 M_0^3 T^2. 
\end{multline}

\subsection{Computation of the Greek coefficients}
\label{subsection:CompGreeks}

This subsection is devoted to the computation of the correction terms in the second order expansion of Theorem \ref{ApproxFormulaProbaApproa2order}.
The analogous derivation for the third order expansion is postponed to Appendix \ref{DerivThirdOrderFormula}.
The correction terms are expressed in terms of Greeks of the payoff function around the proxy model, recall Definition \ref{def:greeks}, and we provide below a useful lemma for their computation.

\begin{lemma}\label{lem3ndOrderExpansionFormula}
Let $\theta$ be a continuous (or piecewise continuous) function and $f$ be a function satisfying Assumption $(\mathbb{S}_1)$.
Then it holds
\begin{align}
\E\left[ \bar{f}\left(\int_0^T \lambda_t\dW_t^L\right) \int_0^T \xi_t\theta_t\dt \right]
	&= \omega(\alpha, \theta)_0^T g_0^f(Y_T^0) +   \omega(\lambda^2, \theta)_0^T g_1^f(Y_T^0), \label{equality1}
	\\ 
\E\left[ \bar{f}\left(\int_0^T \lambda_t\dW_t^L\right) \int_0^T \gamma_t\theta_t\dt \right]
	&= \omega(\beta, \theta)_0^T g_0^f(Y_T^0) +  \rho\omega(\sigma\lambda, \theta)_0^T g_1^f(Y_T^0), \label{equality2}
	\\ 
\E\left[ \bar{f}\left(\int_0^T \lambda_t\dW_t^L\right) \int_0^T Y_{1,t}\theta_t\dt \right]
	&= \left[\omega(\alpha, \alpha_y,\theta)_0^T+\omega(\beta, \alpha_z,\theta)_0^T \right] g_0^f(Y_T^0) \nonumber \\
	&\quad + \left[ \omega(\lambda^2, \alpha_y, \theta)_0^T + \omega(\alpha, \lambda_y\lambda, \theta)_0^T \right] g_1^f(Y_T^0) \label{equality3} \\\nonumber 
	&\!\!\!\!\!\!\!\!\!
	+ \omega(\lambda^2, \lambda_y\lambda, \theta)_0^T g_2^f(Y_T^0)+\rho \omega(\sigma \lambda, \alpha_z,\theta)_0^T g_1^f(Y_T^0), \\ 
\E\left[ \bar{f}\left(\int_0^T \lambda_t\dW_t^L\right) \int_0^T Z_{1,t}\theta_t\dt \right]
	&= \omega(\beta, \beta_z,\theta)_0^T g_0^f(Y_T^0) \nonumber \\
	&\quad + \rho \left[ \omega(\sigma\lambda, \beta_z, \theta)_0^T + \omega(\beta, \lambda \sigma_z, \theta)_0^T \right] g_1^f(Y_T^0) 
	\label{equality4}\\ \nonumber
	&\quad + \rho^2 \omega(\sigma \lambda, \lambda \sigma_z, \theta)_0^T g_2^f(Y_T^0) ,
	\\ 
\E\left[ \bar{f}\left(\int_0^T \lambda_t\dW_t^L\right) \int_0^T \xi_t^2\theta_t\dt \right]
	&= \left[ \omega(\lambda^2, \theta)_0^T + 2\omega(\alpha, \alpha, \theta)_0^T \right] g_0^f(Y_T^0) \nonumber \\
	&\quad + 2 \left[ \omega(\lambda^2, \alpha,\theta)_0^T + \omega(\alpha, \lambda^2,\theta)_0^T\right] g_1^f(Y_T^0) \label{equality5} \\ \nonumber
	&\quad + 2 \omega(\lambda^2, \lambda^2,\theta)_0^T g_2^f(Y_T^0), 
	\\ 
\E\left[ \bar{f}\left(\int_0^T \lambda_t\dW_t^L\right) \int_0^T \gamma_t^2\theta_t\dt \right]
	&= \left[ \omega(\sigma^2, \theta)_0^T + 2\omega(\beta, \beta, \theta)_0^T \right] g_0^f(Y_T^0) \nonumber\\
	&\quad + 2 \rho \left[ \omega(\sigma \lambda, \beta,\theta)_0^T + \omega(\beta, \sigma\lambda,\theta)_0^T\right] g_1^f(Y_T^0) \label{equality6} \\ \nonumber
	&\quad + 2 \rho^2\omega(\sigma \lambda, \sigma \lambda,\theta)_0^T g_2^f(Y_T^0), 
	\\ 
\E\left[ \bar{f}\left(\int_0^T \lambda_tdW_t^L\right) \int_0^T \xi_t\gamma_t\theta_tdt \right]
	&= \left[ \omega(\alpha, \beta ,\theta)_0^T + \omega(\beta, \alpha, \theta)_0^T \right] g_0^f(Y_T^0)\nonumber \\
	&\quad +  \left[ \omega(\lambda^2, \beta,\theta)_0^T + \omega(\beta, \lambda^2,\theta)_0^T\right] g_1^f(Y_T^0) \label{equality7} \\\nonumber
	&\!\!\!\!\!\!\!\!\!\!\!\!\!\!\!\!\!\!\!\!\!\!\!\!\!\!\!
	 +  \rho \Big[ \omega(\lambda \sigma, \theta)_0^T g_0^f(Y_T^0) + \left( \omega(\sigma \lambda, \alpha,\theta)_0^T + \omega( \alpha,\sigma \lambda,\theta)_0^T  \right) g_1^f(Y_T^0)  \nonumber\\ 
	&\quad + \left( \omega( \lambda^2, \sigma \lambda,\theta)_0^T + 
	\omega( \sigma \lambda, \lambda^2,\theta)_0^T  \right) g_2^f(Y_T^0) \Big], \nonumber
\end{align}
where $\bar{f}(x) = f\big( y_0 + \int_0^T \alpha_t\dt + x \big)$, while the processes $\xi$ and $\gamma$ are defined in \eqref{eq:xi} and \eqref{eq:gamma} respectively.
\end{lemma}

\begin{proof}
The equalities are derived by laborious calculations using It\^o's formula and by successively applying Lemmata \ref{IBP1} and \ref{IBP2}.
The details are omitted for the sake of brevity.  
\end{proof}

\subsubsection{Greek coefficients for the second order approximation}

The correction term for the second order expansion is provided by $\E[h^{(1)}(Y_T^0)Y_{1,T}]$ in \eqref{SdOrderExpPayoffanderror} and our target now is to make this explicit. 
Let us recall equations \eqref{proxy0}--\eqref{SDEThirdDerivative}, that $Y_{1,t}=\frac{\partial Y_t^\eta}{\partial\eta}|_{\eta=0}$ and $Z_{1,t}=\frac{\partial Z_t^\eta}{\partial\eta}|_{\eta=0}$ and Remark \ref{rem:notation-coeffs}, which together yield that
\begin{eqnarray}
Y_{1,T}  &=& \int_0^T \xi_t\alpha_{y,t} \dt + \int_0^T \gamma_t \alpha_{z,t} \dt + \int_0^T \xi_t \lambda_{y,t} \dW_t^L,\\
Z_{1,T}  &=& \int_0^T \gamma_t\beta_{z,t} \dt + \int_0^T \gamma_t \sigma_{z,t} \dW_t^X,\\
\xi_t 	 &=& Y^0_t - y_0 = \int_0^t \alpha_s \ds + \int_0^t \lambda_s \dW_t^L, \label{eq:xi} \\
\gamma_t &=& Z^0_t - z_0 = \int_0^t \beta_s \ds + \int_0^t \sigma_s \dW_t^X. \label{eq:gamma}
\end{eqnarray}
Let us also introduce the shifted payoff function
\begin{equation}
\bar{h}^{(i)}(x) = h^{(i)}\left(y_0 + \int_0^T \alpha_t\dt + x\right), \, \text{ for }i \in \N.
\end{equation}
Then we have that
\begin{multline}  \label{2ndOrderExplicit}
\E\left[h^{(1)}(Y_T^0)Y_{1,T}\right] 
	= \E\left[ \bar{h}^{(1)}\left(\int_0^T\lambda_t \dW^L_t\right) \int_0^T \xi_t\alpha_{y,t} \dt \right] \\
	+ \E\left[\bar{h}^{(1)}\left(\int_0^T\lambda_t \dW^L_t\right) \int_0^T \gamma_t\alpha_{z,t} \dt \right] 
	+ \E\left[\bar{h}^{(1)}\left(\int_0^T\lambda_t \dW^L_t\right) \int_0^T \xi_t\lambda_{y,t} \dW^L_t \right]. 
\end{multline}
By applying Lemmata \ref{IBP2} and \ref{lem3ndOrderExpansionFormula}, we obtain
\begin{align}
\E\left[\bar{h}^{(1)}\left(\int_0^T\lambda_t \dW^L_t\right) \int_0^T \xi_t\alpha_{y,t} \dt \right] 
	&= \omega(\alpha, \alpha_y)_0^T g_1^h(Y_T^0) + \omega(\lambda^2, \alpha_y)_0^T g_2^h(Y_T^0), \\
\E\left[\bar{h}^{(1)}\left(\int_0^T\lambda_t \dW^L_t\right) \int_0^T \gamma_t\alpha_{z,t} \dt \right] 
	&= \omega(\beta, \alpha_z)_0^T g_1^h(Y_T^0) + \rho \omega(\sigma \lambda, \alpha_z)_0^T g_2^h(Y_T^0), \\
\E\left[\bar{h}^{(1)}\left(\int_0^T\lambda_t \dW^L_t\right) \int_0^T \xi_t\lambda_{y,t} \dW^L_t \right] 
	&= \omega(\alpha, \lambda_y\lambda)_0^T g_2^h(Y_T^0) + \omega(\lambda^2, \lambda_y \lambda)_0^T g_3^h(Y_T^0).
  \end{align}
More specifically, the first equality follows directly by \eqref{equality1} and the second one by \eqref{equality2}.
For the third equality, we apply first Lemma \ref{IBP2} and then \eqref{equality1}.

Finally, by gathering all the terms, passing to the initial parameters via \eqref{OriginalParameters} below, and writing them as a second order polynomial in $\rho$, we arrive at \eqref{2ndOrderExpansionFormula}.
\begin{eqnarray}\label{OriginalParameters}
\left\{
\begin{array}{lll}
\omega(\alpha, \alpha_y)_0^T 
	& = &  \frac{1}{2}\omega(\lambda^2, \lambda_y\lambda )_0^T  + \rho \left[ \frac{1}{2}\omega(\lambda^2, \lambda_y\sigma)_0^T 
	+ \omega(\lambda \sigma, \lambda_y\lambda)_0^T \right] + \rho^2 \omega(\lambda \sigma, \lambda_y\sigma)_0^T,\\  
\omega(\beta, \alpha_z)_0^T 
	& = & \frac{1}{2} \rho \omega(\sigma^2, \lambda\sigma_z )_0^T,\\
\omega(\lambda^2, \alpha_y)_0^T 
	& = & - \omega(\lambda^2, \lambda_y \lambda)_0^T - \rho \omega(\lambda^2, \lambda_y\sigma)_0^T,\\
\omega(\alpha, \lambda_y \lambda)_0^T 
	& = & - \frac{1}{2} \omega(\lambda^2, \lambda_y \lambda)_0^T - \rho \omega(\lambda \sigma , \lambda_y \lambda)_0^T,\\
\omega(\lambda \sigma, \alpha_z)_0^T 
	& = & -\rho \omega(\sigma \lambda, \lambda\sigma_z )_0^T. 
\end{array}
\right.
\end{eqnarray}

\section{Numerical experiments}\label{NumericalExperiments}

This section is dedicated to numerical experiments and a comparison of the second and third order expansions with the ``market'' approximation for quanto options.

\subsection{Time-homogeneous hyperbolic local volatility model}

We consider the time-homogeneous hyperbolic local volatility model where the SDEs for the forward LIBOR and the forward FX rate are provided by \eqref{equ1a} and \eqref{equ5}, while the coefficients $\lambda(\cdot,y)$ and $\sigma(\cdot,z)$ are homogeneous in time and take the form:
\begin{align}
\lambda(y)
	&:= \nu_L \left[ \frac{1-\beta_{L} + \beta^2_{L}}{\beta_{L}} + \frac{(\beta_L-1)}{\beta_L} 
			\left( \frac{\sqrt{y^2 + \beta_L^2(1-y)^2} - \beta_L}{y} \right)  \right], \\
\sigma(z)
	&:= \nu_X \left[ \frac{1-\beta_{X} + \beta^2_{X}}{\beta_{X}} + \frac{(\beta_X-1)}{\beta_X}
			\left( \frac{\sqrt{z^2 + \beta_X^2(1-z)^2} - \beta_X}{z} \right)  \right],
\end{align}
where $\nu_L$ and $\nu_X$, both strictly positive, represent the levels of volatility, while $\beta_L$ and $\beta_X$, both valued in $[0,1]$, represent the skew 
parameters. 
This model corresponds to the Black--Scholes model for $\beta_L = \beta_X = 1$ and exhibits a skew for the \text{imp}lied volatility surface when $\beta_L \, \text{or} \, \beta_X  \neq 1$. 
It was introduced by \citet{Jac08}, behaves similarly to the CEV (Constant Elasticity of Variance) model, and has been used for numerical experiments also in \citet{BompisHok14}. 
The advantage of this model is that zero is not an attainable boundary, and that allows to avoid some numerical instabilities present in the CEV model when the 
underlying LIBOR or FX rate are close to zero; see \textit{e.g.} \citet{And00}.
Although the assumptions of boundedness and ellipticity are not fulfilled, we reasonably expect that our approximation formulas remain valid for this model, and apply Theorems \ref{ApproxFormulaProbaApproa2order} and \ref{ApproxFormulaProbaApproa3order}. 
The payoff of a call option on the other hand does satisfy the smoothness assumption $(\mathbb S_2)$, as the payoff is everywhere differentiable apart from the kink at the strike level and grows exponentially.
The numerical experiments that follow show that the derived approximations perform well in this setting, even though some theoretical assumptions are not satisfied.

\subsection{Market approximation for pricing of European quanto option}
\label{commonmktappro}

The common market practice is to evaluate European quanto call/put options analytically using a Black--Scholes type formula with a quanto drift correction. 
More precisely, for a caplet with maturity date $T$, strike $K$ and payment date $T_1$, the market approximation is provided by
\begin{equation}\label{MarketApproximation}
\mathbb C^{M}(T,K) = \delta B(0,T_1)\left( \e^{y_0 - q T} \Phi(d_1) - \e^k \Phi(d_2) \right),
\end{equation}
where $q = \rho \lambda_{\text{imp}}(T, \text{ATM})\sigma_{\text{imp}}(T, \text{ATM})$, $k=\ln K$, $\Phi$ is the cdf of the standard normal distribution, and
\begin{equation}
d_1 = \frac{y_0-k-\rho \lambda_{{\text{imp}}}(T,\text{ATM})\sigma_{\text{imp}}(T,\text{ATM})T + \frac{1}{2}\lambda_{\text{imp}}(T,k)^2T}{\lambda_{\text{imp}}(T,k)\sqrt{T}},
\end{equation}
\begin{equation}
 d_2 = d_1-\lambda_{\text{imp}}(T,k)\sqrt{T},
\end{equation}
where $\lambda_{\text{imp}}(T,\text{ATM})$, $\sigma_{\text{imp}}(T,\text{ATM})$ are respectively the \text{ATM} \text{imp}lied volatility for the forward LIBOR rate and the FX forward rate with expiry $T$, while $\lambda_{\text{imp}}(T,k)$  is \text{imp}lied volatility for the forward LIBOR rate with strike $k$.
This approach is similar to the ``practitioner'' Black--Scholes model considered in \citet{ChrisJac04} or \citet{Romo12}.
Observe that the approximation formula \eqref{MarketApproximation} becomes exact by construction when $\rho=0$.

\subsection{Comparison results for the second and third order expansions and the market approximation} 

\subsubsection{Set of parameters} 

The numerical experiments are conducted using the following values for the parameters: $L_0=6\%$, $X_0=1$, $\nu_L=8\%$, $\beta_L = 0.3$, $\nu_X=15\%$ and $\beta_X = 0.5$. 
They are chosen to be comparable to market values, see \textit{e.g.} \citet{Hullwhite00} and \citet{NgSun08}.
In order to illustrate this, Figures \ref{GraphLIBORImpliedVol} and \ref{GraphFXImpliedVol} show, respectively, the implied volatilities for the forward LIBOR and the forward FX rates generated with these parameters for various maturities. 
They represent the skew typically observed in interest rates and FX markets.

The challenging part for the pricing comes from the choice of the correlation parameter $\rho$ between the forward LIBOR rate and the foreign exchange rate, because its level is not directly observable in the market and has a significant impact in the pricing as showed in Figure \ref{GraphCorrelationImpact}. 
In practice, its level is either chosen by the trader or estimated using historical data. 
The empirical analysis in \citet{BoenSch03} shows that the estimated correlations depend on the underlying interest rates and the pair of currencies considered. 
In general, $\rho$ is not too {\it{large}} and belongs to the region $[-0.2, 0.2]$. 
A trader who sells this product, may choose its level in a conservative way (higher selling price) by taking a lower or negative correlation level, as the option price is decreasing with $\rho$. 
For these reasons and for the purpose of testing our formulas, we consider correlation levels $\rho \in [-0.5, 0.5]$.

In order for the tests to be comprehensive, we consider various relevant maturities (1, 6, 10 and 15 years) and strikes (with a range increasing with the maturity). 
This roughly covers very out-of-the-money options and very in-the-money options.

\subsubsection{Benchmarks} 

Benchmarks for model prices are computed using the Monte Carlo method by discretizing the diffusion process using the Euler scheme. 
The number of Monte Carlo (MC) paths and the number of steps in the discretization are chosen such that the $95\%$ confidence intervals are within $2$ basis points.

\subsubsection{Accuracy}

The results for the tests are illustrated in Figures \ref{GraphErrorRhominuszerofive}, \ref{GraphErrorRhominuszerotwo}, \ref{GraphErrorRhopluszerotwo} and 
\ref{GraphErrorRhopluszerofive}.
The following observations stem from these tests and their illustrations: \medskip

\noindent In general, the test results up to 15 years show that the second and third order approximation formulas provide very good accuracy. 
Table \ref{tab:statisticsSdTdApproximation} gives some statistics (average and maximum for the absolute discrepancy) for various correlation values considered. 
The maximum average error for the second (third) order approximation formulas is $2.8$ ($2.2$) bps with correlation value equal to $-0.5$.
The maximum error for the second (third) order approximation is about $14$ ($8.4$) bps. 
Third order approximation formulas produce better accuracy in comparison to the second order approximation formulas, which is expected. \\[-.75em]

\noindent The market approximation formulas provide good accuracy as well; see the statistics in Table \ref{tab:statisticsMarketApproximation}.   
This is due to the fact that this formula is exact for $\rho=0$, which results in good accuracy when the correlation parameter is fairly small.
Indeed, the largest average error for the market approximation formulas is $3.3$ bps with correlation value equal to $-0.5$.
The maximum error for the market approximation is about $12.6$ bps.
\\[-.75em]

\noindent In order to compare the different methods, let us mention that when the impact of the quanto effect becomes important (\textit{i.e.} $\rho= \pm 0.5$), the accuracy of the third order approximation is better than the one given by the market approximation (see Figures \ref{GraphErrorRhominuszerofive} and \ref{GraphErrorRhopluszerofive}), whereas the precision given by the second order and the market approximations is comparable. 
For a reduced quanto effect (\textit{i.e.} $\rho=\pm 0.2$), the accuracy from the third order and the market approximation is similar. 
Indeed, in the limiting case of $\rho \rightarrow 0$, the market approximation becomes exact by construction. 
The main advantage of our expansion formulas is to provide an accurate estimation of the error which is directly related to the maturity of the option ($T$), the level and curvature of the local volatility functions ($M_0$ and $M_1$) and the quanto impact ($\rho$).

\begin{figure}[htbp]
  \centering
  \includegraphics[width=\textwidth]{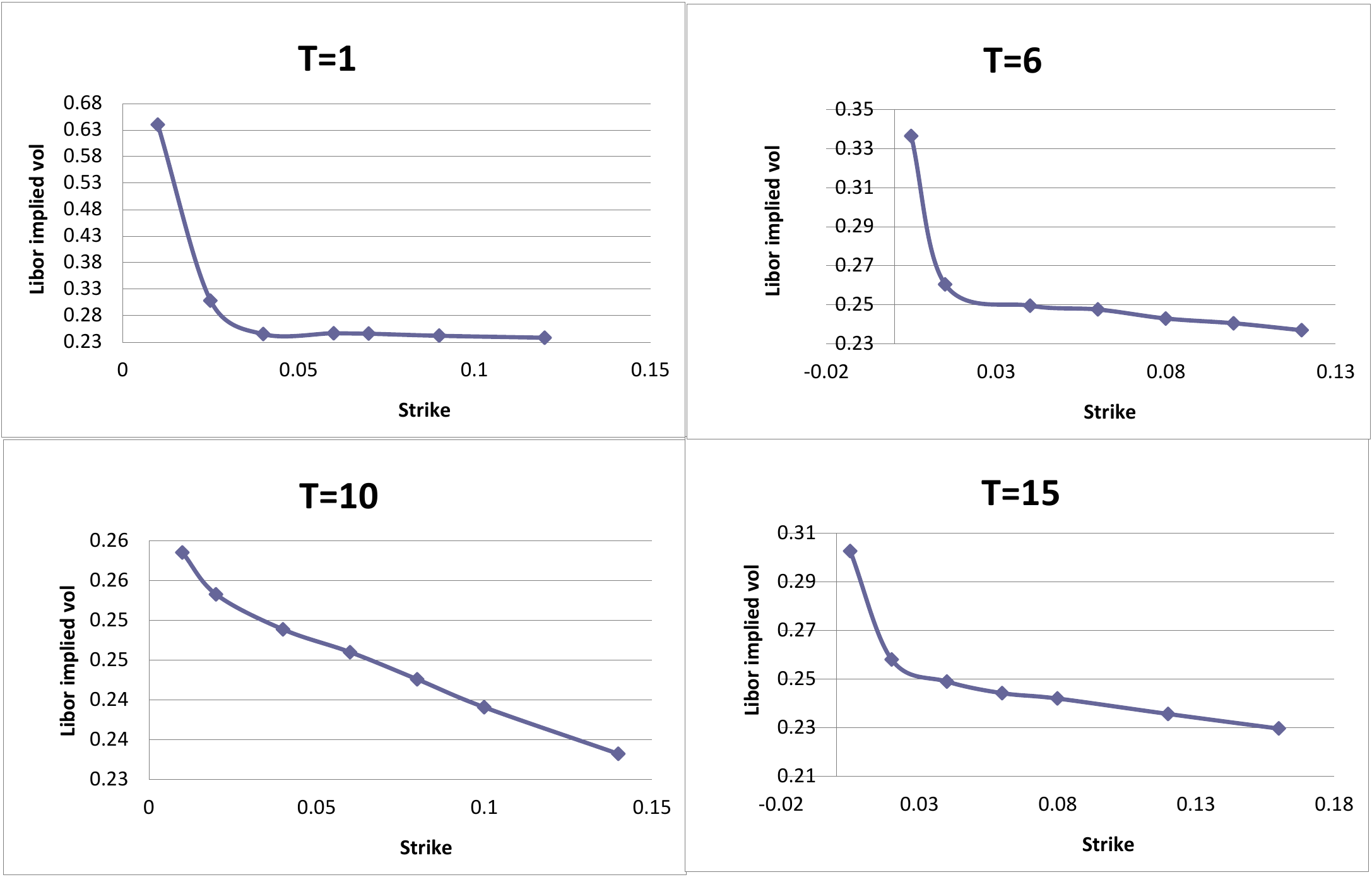}
  \caption{Forward LIBOR rate implied volatility generated with parameters $L_0=6\%$, $\nu_L = 8\%$, $\beta_L = 0.3$ for various maturities.}
  \label{GraphLIBORImpliedVol}
\end{figure}

\begin{figure}[htbp]
  \centering
  \includegraphics[width=\textwidth]{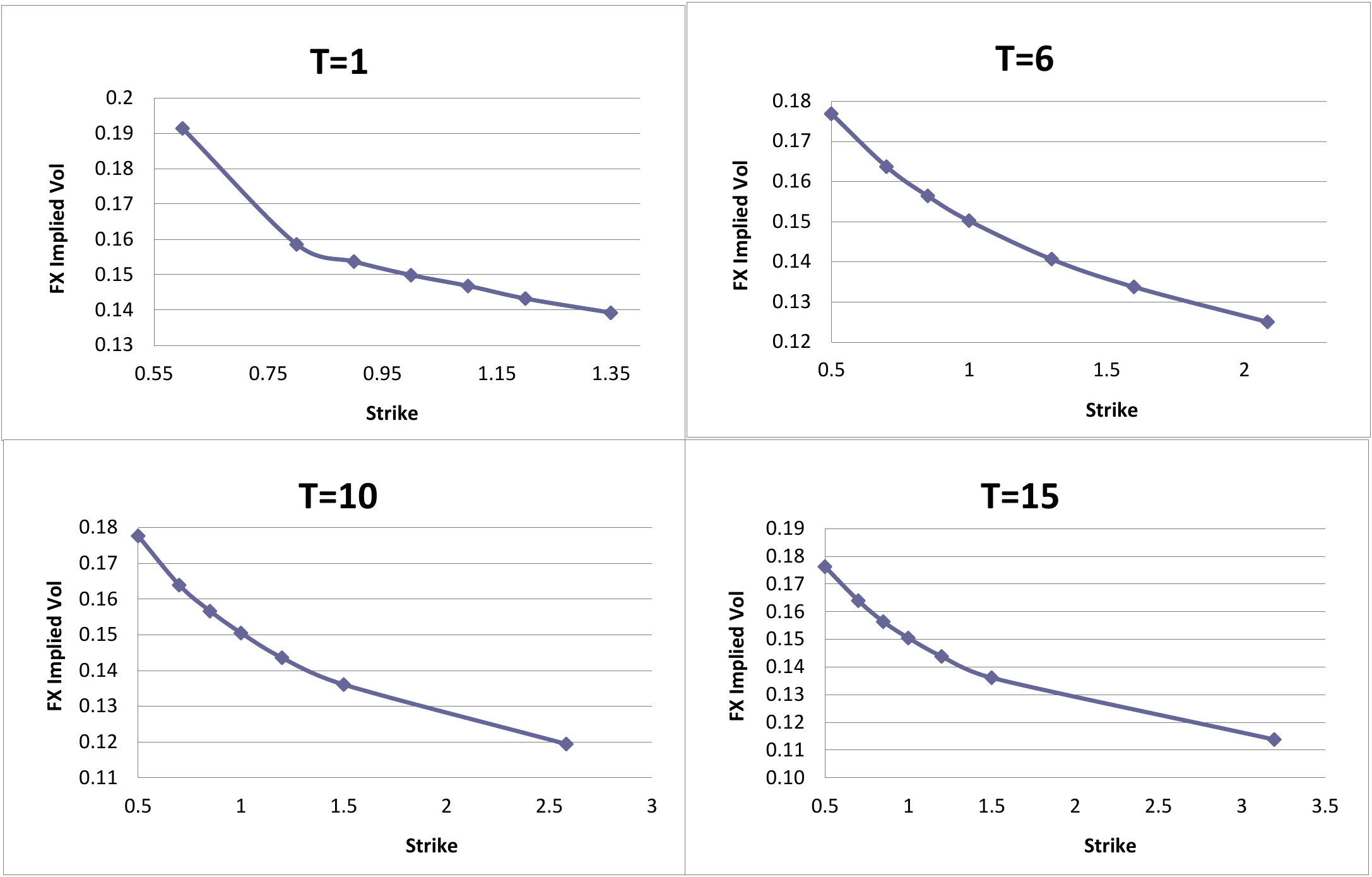}
  \caption{FX forward rate implied volatility generated with parameters $X_0=1$, $\nu_X = 15\%$, $\beta_X = 0.5$ for various maturities.}
  \label{GraphFXImpliedVol}
\end{figure}

\begin{figure}[htbp]
  \centering
  \includegraphics[width=1\textwidth]{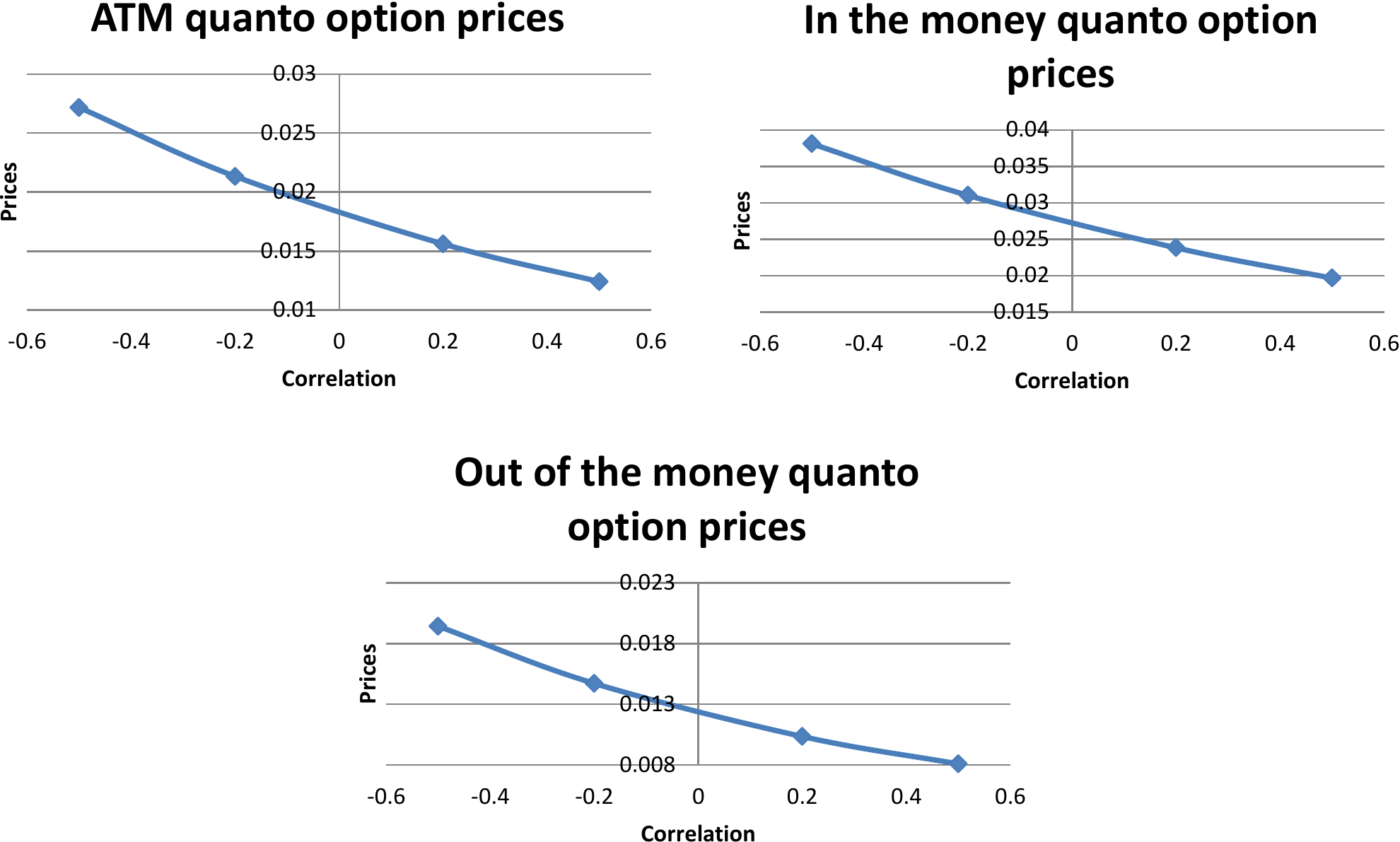}
  \caption{Impact of the correlation parameter $\rho$ for in-the-money ($K = 4\%$), \text{ATM} ($K=6\%$) and out-of-the-money ($K = 8\%$) option prices.}  
  \label{GraphCorrelationImpact}
\end{figure}

\begin{figure}[htbp]
  \centering
  \includegraphics[width=1.0\textwidth]{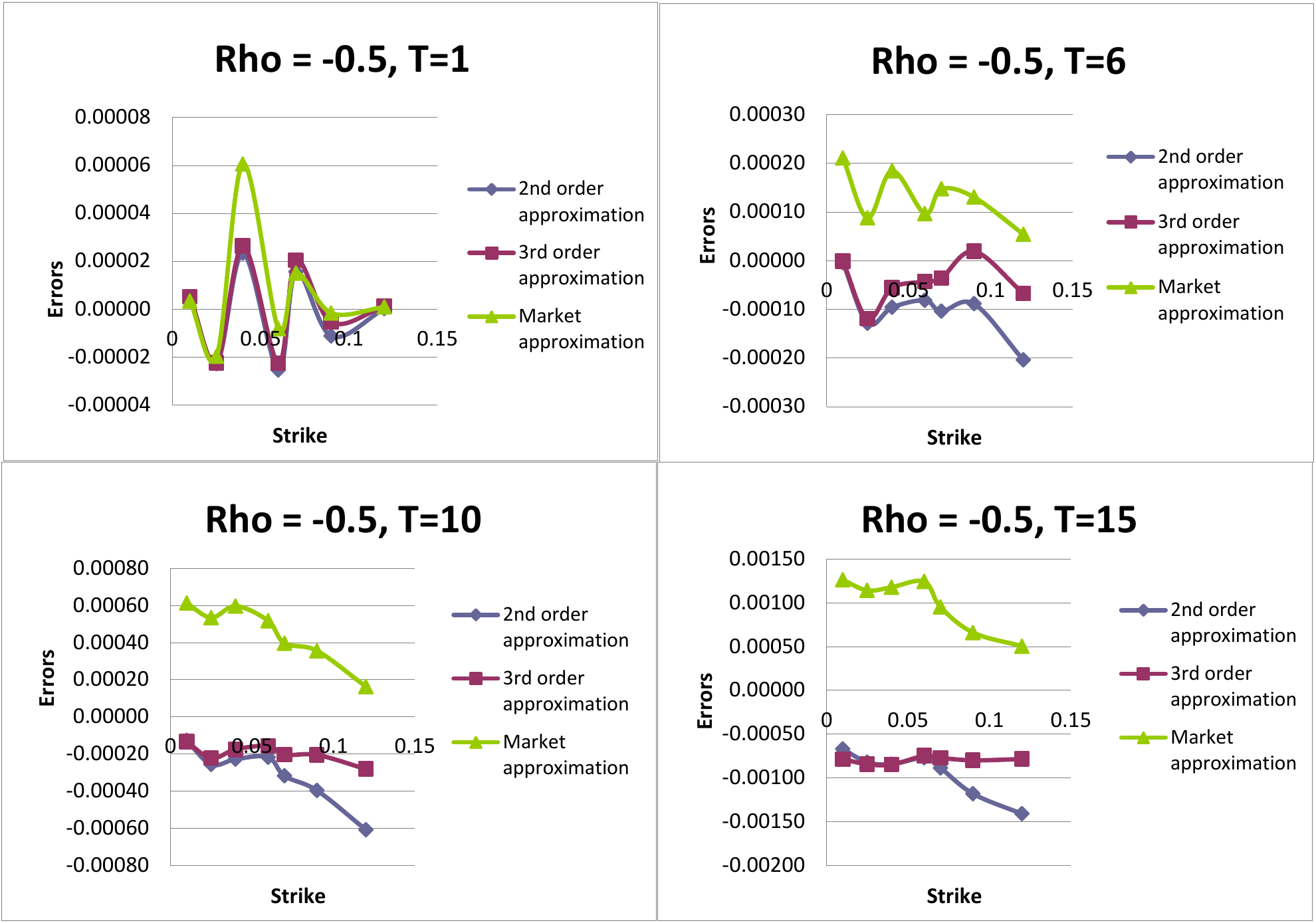}
  \caption{Absolute discrepancy between the benchmarks prices and those 
calculated with different approximation schemes when $\rho=-0.5$.}  
  \label{GraphErrorRhominuszerofive}
\end{figure}

\begin{figure}[htbp]
  \centering
  \includegraphics[width=1.0\textwidth]{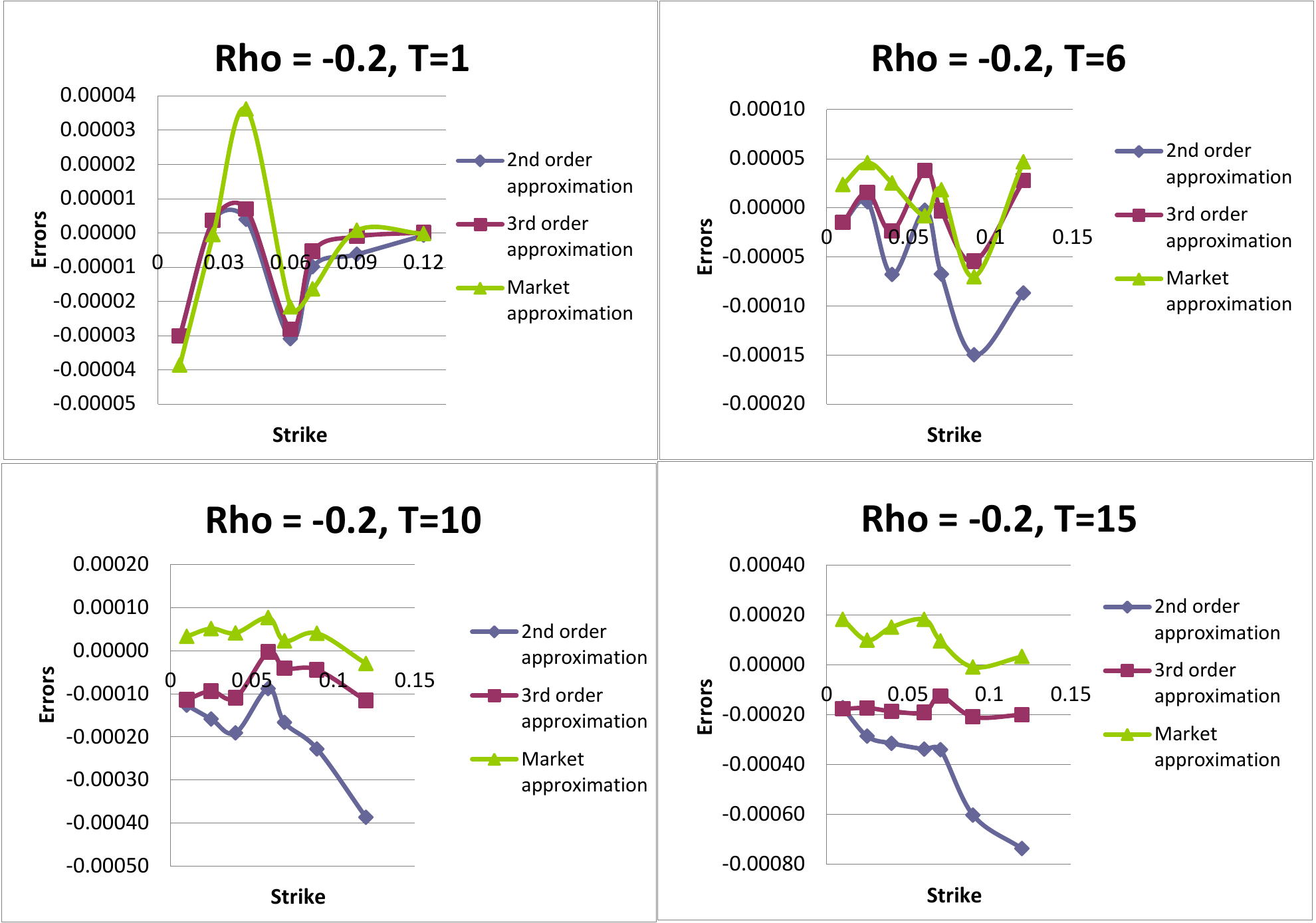}
  \caption{Absolute discrepancy between the benchmarks prices and those 
calculated with different approximation schemes when $\rho=-0.2$. } 
  \label{GraphErrorRhominuszerotwo}
\end{figure}

\begin{figure}[htbp]
  \centering
  \includegraphics[width=1.0\textwidth]{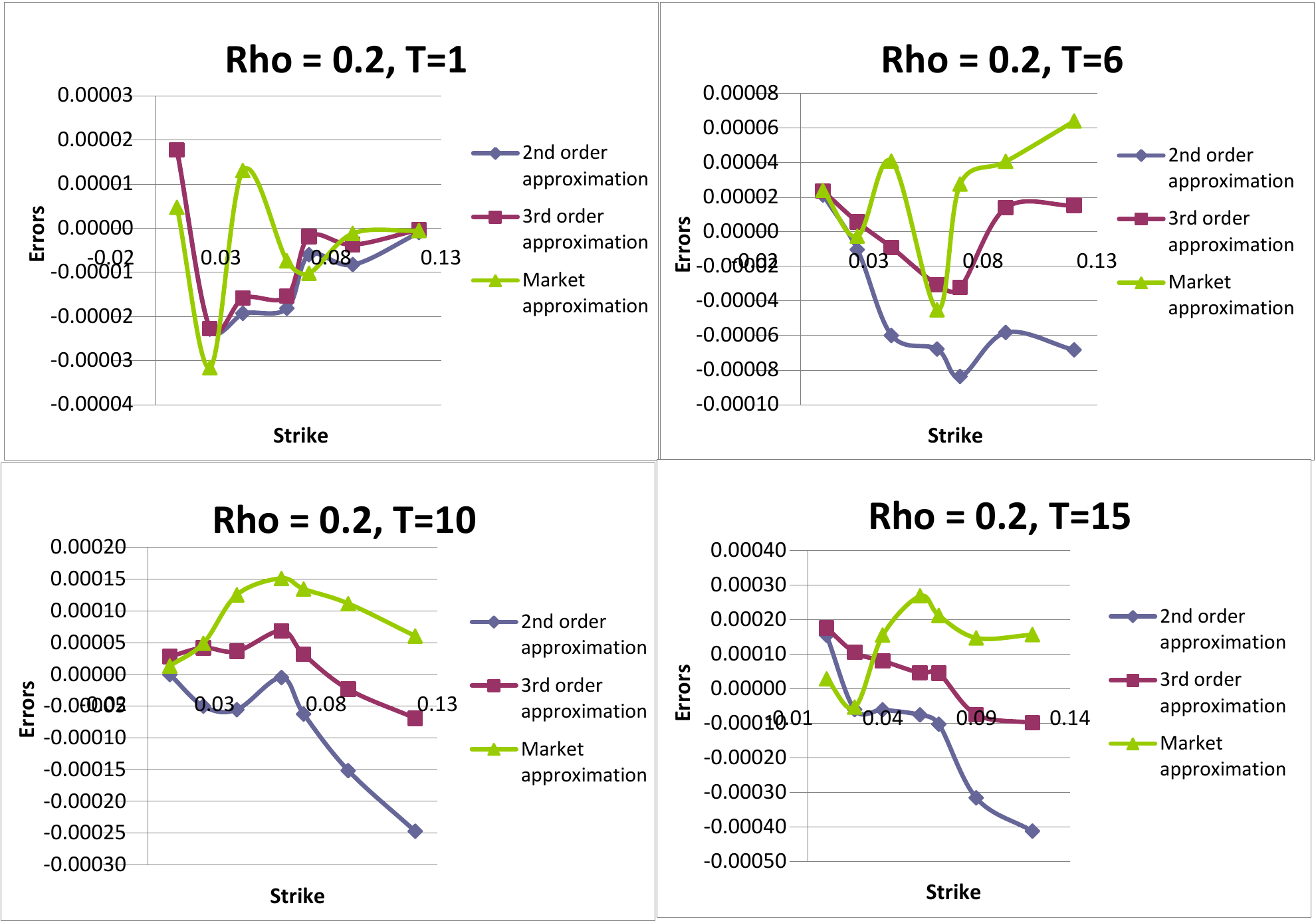}
  \caption{Absolute discrepancy between the benchmarks prices and those 
calculated with different approximation schemes when $\rho=0.2$.}  
  \label{GraphErrorRhopluszerotwo}
\end{figure}

\begin{figure}[htbp]
  \centering
  \includegraphics[width=1.0\textwidth]{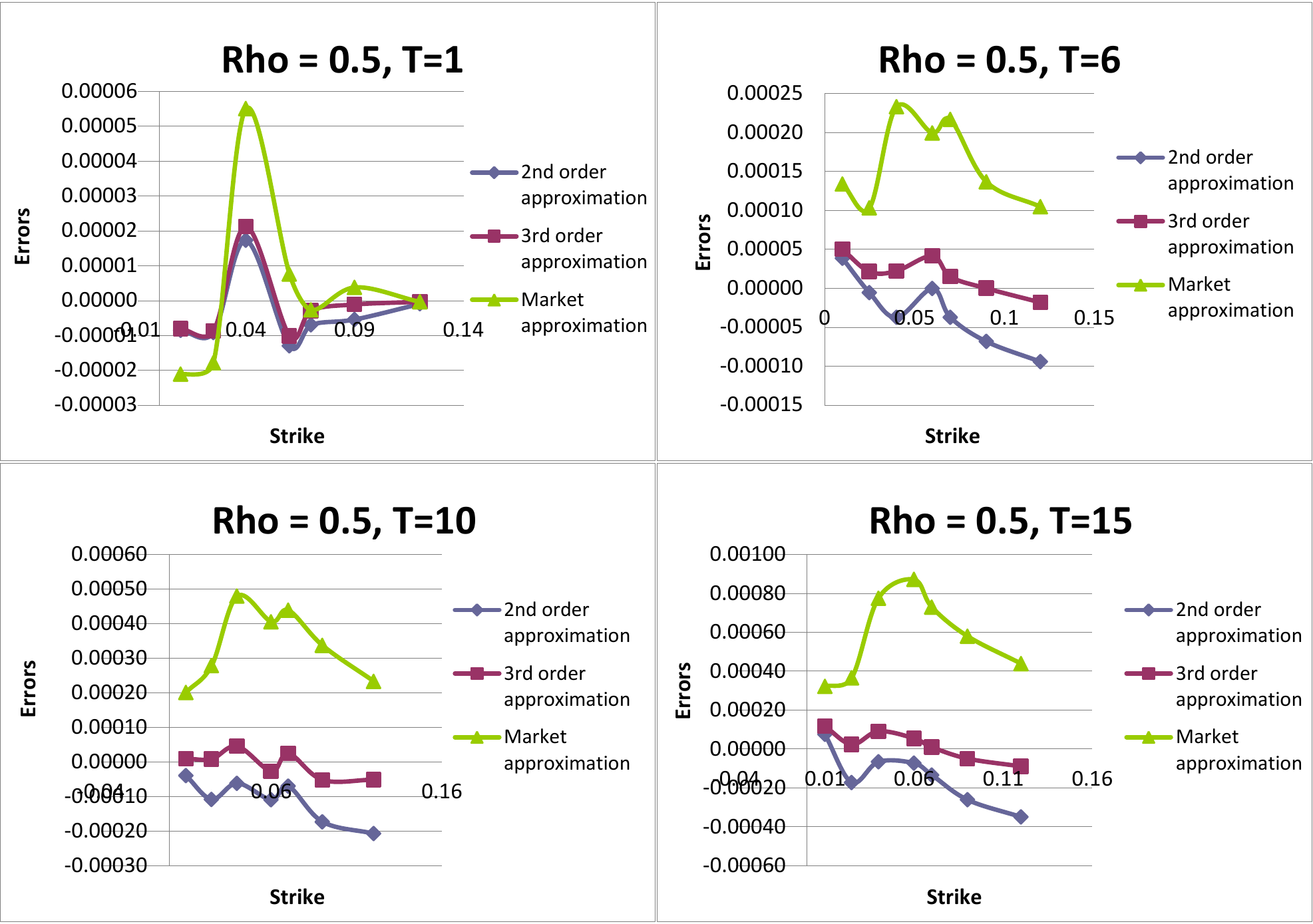}
  \caption{Absolute discrepancy between the benchmarks prices and those 
calculated with different approximation schemes when $\rho=0.5$.}  
  \label{GraphErrorRhopluszerofive}
\end{figure}

\begin{table}[htbp]
  \centering
    \begin{tabular}{ccccc}
    \hline
    Correlation & Average (2nd order) & MAX (2nd order) & Average (3rd order) & MAX(3rd order) \\
   \hline
    -0.5  & 0.00028 & 0.00141 & 0.00022 & 0.00084 \\
    -0.2  & 0.00014 & 0.00074 & 0.00006 & 0.0002 \\
    0.2   & 0.00007 & 0.00041 & 0.00004 & 0.00017 \\
    0.5   & 0.00007 & 0.00035 & 0.00003 & 0.00011 \\
   \hline
    \end{tabular}\\[.35em]
  \caption{Average and maximum statistics for the absolute discrepancy, for various correlation values considered}
  \label{tab:statisticsSdTdApproximation}%
\end{table}%

\begin{table}[htbp]
  \centering
    \begin{tabular}{ccc}
    \hline
    Correlation & Average (market approximation) & MAX (market approximation) \\
    \hline
    -0.5  & 0.00033 & 0.00126 \\
    -0.2  & 0.00005 & 0.00018 \\
    0.2   & 0.00006 & 0.00027 \\
    0.5   & 0.00023 & 0.00087 \\
    \hline
    \end{tabular}\\[.35em]
  \caption{Average and maximum statistics for the absolute discrepancy, for various correlation values considered}
  \label{tab:statisticsMarketApproximation}%
\end{table}%


\appendix
\label{sec:appendix}

\section{Proof of Lemma \ref{deryBSlemma}}\label{app:proof}

\begin{proof}
Let us write 
\begin{equation}
\mathbb C^{BS}(y_0) = \e^{-\Sigma(T)} \tilde{\mathbb C}^{BS}(y_0)
\end{equation}
where 
\begin{equation}
\tilde{\mathbb C}^{BS}(y_0) = \e^{y_0} \Phi(d_1) - \e^{\tilde{k}} \Phi(d_2)
\end{equation}
with 
\begin{equation}
\tilde{k} = k+\Sigma(T), \, \, d_1 = \frac{y_0-\tilde{k}+\frac{1}{2} \Lambda(T)}{\sqrt{\Lambda(T)}}, \,d_2 = d_1 - \sqrt{\Lambda(T)}.	
\end{equation}
For $n=1$, we get $\frac{\partial}{\partial y} \tilde{\mathbb C}^{BS}(y_0) = \e^{y_0} \Phi(d_1)$.
For $n \geq 2$, we apply the Leibniz formula for the product $\e^{y_0} \Phi(d_1)$. 
\end{proof}

\section{Malliavin calculus}

We start by introducing some definitions and notation for the Malliavin calculus --- see \textit{e.g.} \citet{BallyCaraLombl10} or \citet{Nualart05} more for details --- before providing two lemmas for the $L^p$ estimates of Malliavin derivatives and the integration-by-parts formulas. 

Let us write $W^L = \rho W^X + \sqrt{1-\rho^2} \widetilde{W}^L$, where $(\widetilde{W}_t^L)_{0 \leq t \leq T}$ is a Brownian motion independent of 
$(W_{t}^X)_{0 \leq t \leq T}$, and consider the Malliavin calculus for the 2-dimensional Brownian motion $(\widetilde{W^L}, W^X)$. 
Let $D^i_t F, \, i=1,2, $ denote the Malliavin derivative of the random variable $F$ wrt to the Brownian motion $i$ at time $t$, and similarly for the higher order derivatives, where for example $D^{i,j}_{t_1, t_2}F = D^i_{t_1}D^j_{t_2} F$.

Under the regularity assumption $(\mathbb R_4)$, using \cite{Nualart05}, we know that for any $t \leq T$, any $\eta \in [0,1]$ and any $p \geq 1$, we have $(Y_t^\eta,Z_t^\eta) \in \D^{4,p}, (Y_{1,t}^\eta,Z_{1,t}^\eta) \in \D^{3,p}, (Y_{2,t}^\eta,Z_{2,t}^\eta) \in \D^{2,p}$ and $(Y_{3,t}^\eta,Z_{3,t}^\eta) \in \D^{1,p}$. 
The existence of any  moment is easy to establish, see \textit{e.g.} \citet{Priouret05} or \citet{Nualart05}. 
We focus on the Malliavin differentiability of the system of SDEs and their $L^p$ estimates.

Let $r>t$, then $D^1_r Y_t^\eta = D^2_r Y_t^\eta = 0$. 
Now take $r \leq t$, then $(D^1_r Y_t^\eta, D^2_r Y_t^\eta)$ solves the following system of SDEs
\begin{equation}\label{YSDEMalliavinFirstDeriv}
\left\lbrace
\begin{array}{lll}
D^1_r Y_t^\eta 
	&=&  \lambda\big(r, \eta Y_r^\eta+(1-\eta)y_0\big) \sqrt{1-\rho^2} + \int_r^t \eta \alpha_y D^1_r Y_u^\eta \ud u \\
	&& +\int_r^t \eta \lambda_y D^1_r Y_u^\eta (\sqrt{1-\rho^2} \ud \widetilde{W}_t^L + \rho \ud W_t^X),\\
D^2_r Y_t^\eta 
	&=&  \lambda\big(r, \eta Y_r^\eta+(1-\eta)y_0\big) \rho + \int_r^t \eta (\alpha_y D^2_r Y_u^\eta  + \alpha_z D^2_r Z_u^\eta)\ud u \\
	&& +\int_r^t \eta \lambda_y D^2_r Y_u^\eta (\sqrt{1-\rho^2} \ud \widetilde{W}_t^L + \rho \ud W_t^X).
\end{array} \right.
\end{equation}
For the second order Malliavin derivatives, where for instance $r < s \leq t$, we have 
\begin{eqnarray}\label{YSDEMalliavinSecondDeriv}
\left\lbrace
\begin{array}{lll}
D^{1,1}_{r,s} Y_t^\eta 
	&=&  \eta \lambda_y D^1_r Y_s^\eta  + \int_s^t \eta \big(\alpha_y D^{1,1}_{s,r} Y_u^\eta +  \alpha_{yy}D^1_r Y_u^\eta D^1_s Y_u^\eta  \big) \ud u \\
	&&+ \int_s^t \eta \big( \lambda_y D^{1,1}_{s,r} Y_u^\eta + \lambda_{yy} D^1_r Y_u^\eta D^1_s Y_u^\eta \big) (\sqrt{1-\rho^2} \ud \widetilde{W}_t^L + \rho 
		\ud W_t^X),\\
D^{1,2}_{r,s} Y_t^\eta 
	&=&  \eta \lambda_y D^1_r Y_s^\eta +  \int_s^t \eta \big( \alpha_y D^{1,2}_{r,s} Y_u^\eta +  \alpha_{yy}D^1_r Y_u^\eta D^2_s Y_u^\eta 
			+ \alpha_{yz}D^1_r Y_u^\eta D^2_s Z_u^\eta  \big) \ud u\\
	&&+ \int_s^t \eta \big( \lambda_y D^{1,2}_{r,s} Y_u^\eta + \lambda_{yy} D^1_r Y_u^\eta D^2_s Y_u^\eta \big) (\sqrt{1-\rho^2} \ud \widetilde{W}_t^L + \rho 
		\ud W_t^X),\\
D^{2,2}_{r,s} Y_t^\eta 
	&=& \eta \lambda_y D^2_r Y_s^\eta +  \int_s^t \eta \big( \alpha_y D^{2,2}_{s,r} Y_u^\eta +  \alpha_{yy}D^2_r Y_u^\eta D^2_s Y_u^\eta  
		+\alpha_{yz}D^2_r Y_u^\eta D^2_s Z_u^\eta \\
	&&+  \alpha_z D^{2,2}_{s,r} Z_u^\eta +  \alpha_{zy}D^2_r Z_u^\eta D^2_s Y_u^\eta  +\alpha_{zz}D^2_r Z_u^\eta D^2_s Z_u^\eta \big) \ud u\\
	&&+  \int_s^t \eta \big( \lambda_y D^{2,2}_{r,s} Y_u^\eta + \lambda_{yy} D^2_r Y_u^\eta D^2_s Y_u^\eta \big) (\sqrt{1-\rho^2} \ud \widetilde{W}_t^L + \rho 
		\ud W_t^X).
\end{array} 
\right.
\end{eqnarray}

The process $Z^\eta$ is independent from $\widetilde{W}^L$, hence its Malliavin derivatives wrt to it are zero. 
Similarly, for $r>t, \, D^2_r Z_t^\eta = 0$, while for $r \leq t$, $D^2_r Z_t^\eta$ solves
\begin{equation}
D^2_r Z_t^\eta 
	= \sigma(r, \eta Z_r^\eta+(1-\eta)z_0) + \int_r^t \eta \beta_z D^2_r \, Z_u^\eta \ud u + \int_r^t \eta \sigma_z D^2_r Z_u^\eta \ud W_u^X.
\end{equation}
For the second order Malliavin derivatives, take for instance $r < s \leq t$, we have 
\begin{equation}\label{ZSDEMalliavinSecondDeriv}
\left\lbrace
\begin{array}{lll}
D^{2,2}_{r,s} Z_t^\eta 
	&=&   \eta \sigma_z D^2_r Z_s^\eta + \int_s^t \eta \big( \beta_{zz} D^2_s Z_u^\eta   D^2_r Z_u^\eta  + \beta_z D^{2,2}_{r,s} Z_u^\eta\big) \ud u  \\
	&& +  \int_s^t \eta \big( \sigma_{zz} D^2_s Z_u^\eta   D^2_r Z_u^\eta  + \sigma_z D^{2,2}_{r,s} Z_u^\eta\big) \ud W_u^X.
\end{array} \right.
\end{equation}
Similarly, we provide the Malliavin derivatives for $(Y_{1,t}^\eta,Z_{1,t}^\eta)$. $D^{1} Z_{1,t}^\eta=0$ because $Z^\eta$ is independent of $\tilde{W}^L$. 
$D^{2}_{r} Z_{1,t}^\eta=0$ for $r>t$. 
For $r \leq t$  we have
\begin{equation}
\left\lbrace
\begin{array}{lll}
D^{2}_{r} Z_{1,t}^\eta 
	&=& \sigma_z ( \eta Z_{1,r}^\eta +  Z_{r}^\eta-z_0) \\
	&& + \int_r^t \big[ \eta \beta_{zz} D^2_r Z_u^\eta(\eta Z_{1,u}^\eta +  Z_{u}^\eta-z_0) + \beta_{z}(\eta D^2_r Z_{1,u}^\eta + D^2_r Z_{u}^\eta) \big] \ud u\\
	&& + \int_r^t \big[ \eta \sigma_{zz} D^2_r Z_u^\eta(\eta Z_{1,u}^\eta +  Z_{u}^\eta-z_0) + \sigma_{z}(\eta D^2_r Z_{1,u}^\eta + D^2_r Z_{u}^\eta) \big] \ud W_u^X			
\end{array} \right.
\end{equation}
Furthermore, $D^{i}_{r} Y_{1,t}^\eta=0$ for $r>t, \, i=1,2$, while for $r \leq t$, we get
\begin{eqnarray}
\left\lbrace
\begin{array}{lll}
D^{1}_{r} Y_{1,t}^\eta 
	&=&   \lambda_y ( \eta Y_{1,r}^\eta +  Y_{r}^\eta-y_0) 
	   	+ \int_r^t \big[ \eta \alpha_{yy}  D^1_r Y_u^\eta(\eta Y_{1,u}^\eta +  Y_{u}^\eta-y_0) \\
	&&\quad + \alpha_{y}(\eta D^1_r Y_{1,u}^\eta + D^1_r Y_{u}^\eta)
		+\eta \alpha_{zy}  D^1_r Y_u^\eta(\eta Z_{1,u}^\eta +  Z_{u}^\eta-z_0) \big] \ud u  \\
	&&   + \int_r^t \big[ \eta \lambda_{yy}  D^1_r Y_u^\eta(\eta Y_{1,u}^\eta +  Y_{u}^\eta-y_0)   
		+ \lambda_{y}(\eta D^1_r Y_{1,u}^\eta + D^1_r Y_{u}^\eta)  \big] \\
	&&\qquad (\sqrt{1-\rho^2} d \widetilde{W}_t^L + \rho \ud W_t^X),\\
D^{2}_{r} Y_{1,t}^\eta 
	&=& \rho  \lambda_y ( \eta Y_{1,r}^\eta +  Y_{r}^\eta-y_0) \\
	&&   + \int_r^t \big[ \eta(\eta Y_{1,u}^\eta + Y_{u}^\eta-y_0) ( \alpha_{yy}  D^2_r Y_u^\eta + \alpha_{yz}  D^2_r Z_u^\eta)  + \alpha_{y}(\eta D^2_r 
		Y_{1,u}^\eta + D^2_r Y_{u}^\eta) \\
	&&		\eta(\eta Z_{1,u}^\eta + Z_{u}^\eta-z_0) ( \alpha_{yz}  D^2_r Y_u^\eta + \alpha_{zz}  D^2_r Z_u^\eta)  + \alpha_{z}(\eta D^2_r Z_{1,u}^\eta + 
		D^2_r Z_{u}^\eta)	\big] \ud u  \\
	&&   + \int_r^t \big[ \eta \lambda_{yy}  D^2_r Y_u^\eta(\eta Y_{1,u}^\eta +  Y_{u}^\eta-y_0)   			+ \lambda_{y}(\eta D^2_r Y_{1,u}^\eta + D^2_r 
		Y_{u}^\eta)  \big] \\
	&& \qquad (\sqrt{1-\rho^2} \ud \widetilde{W}_t^L + \rho \ud W_t^X).		
					
\end{array} \right.
\end{eqnarray}
Other Malliavin derivatives for this system of SDEs can be derived similarly, without particular difficulties.
Following \citet[Theorem 5.2, Step 2]{MiriGobetBen02}, we provide in the following lemma tight estimates for this system of Malliavin derivatives, which are useful for the error analysis.  

\begin{lemma}[Estimates of Malliavin derivatives]\label{EstimatesMalliavinDerivatives}
The following hold, for any $p \geq 1$ and $i,j,k=1,2$:
\begin{align}
\begin{array}{lll}
\E |D^{i}_rZ^\eta_t|^p \leq_c |\sigma|_{\infty}^p &&
\E |D^{i}_rY^\eta_t|^p \leq_c |\lambda|_{\infty}^p\\
\E |D^{i,j}_{r,s}Z^\eta_t|^p \leq_c |\sigma|_{\infty}^pM_1^p &&
\E |D^{i,j}_{r,s}Y^\eta_t|^p \leq_c |\lambda|_{\infty}^pM_1^p \\
\E |D^{i,j,k}_{r,s,u}Z^\eta_t|^p \leq_c |\sigma|_{\infty}^p M_0^p M_1^p &&
\E |D^{i,j,k}_{r,s,u}Y^\eta_t|^p \leq_c |\lambda|_{\infty}^p M_0^p M_1^p \\
\E |D^{i}_{r}Z_{1,t}^\eta|^p \leq_c M_1^p (M_0 \sqrt{T})^p &&
\E |D^{i}_{r}Y_{1,t}^\eta|^p \leq_c M_1^p (M_0 \sqrt{T})^p \\
\E |D^{i,j}_{r,s}Y_{1,t}^\eta|^p \leq_c M_0^p M_1^p &&
\E |D^{i,j}_{r,s}Z_{1,t}^\eta|^p \leq_c M_0^p M_1^p \\
\E |D^{i}_rZ^\eta_{2,t}|^p \leq_c M_1^p (M_0\sqrt{T})^{2p}&&
\E |D^{i}_rY^\eta_{2,t}|^p \leq_c M_1^p (M_0\sqrt{T})^{2p}\\
\E |D^{i}_rZ^\eta_{3,t}|^p \leq_c M_1^p (M_0\sqrt{T})^{3p}&&
\E |D^{i}_rY^\eta_{3,t}|^p \leq_c M_1^p (M_0\sqrt{T})^{3p}
\end{array}
\end{align}
uniformly in $(r,s,t,u) \in [0,T]$ and $\eta \in [0,1]$.
\end{lemma}

In the following lemma, we state a key result of the integration-by-parts formula allowing to represent the error term \eqref{error2} using only $h^{(1)}$ and 
providing some moments' control useful in the error analysis.  

\begin{lemma}\label{ErrorAnalysisIBP} 
Let Assumptions $\mathbb{(ELL)}, \mathbb{(RHO)}$ and $(\mathbb R_3)$ be in force. 
Let $Z$ belong to $\cap_{p \geq 1} \D^{2,p}$. 
Then, for any $\eta \in [0,1]$, for $k=1,2$, there exists a random variable $Z_k^\eta$ in any $L^p (p\geq 1)$ such that for any function $l \in {C}_0^{\infty}(\R,\R)$, one has
\begin{equation}
\E\big[ l^{(k)}(\eta Y_T + (1-\eta)Y_T^0) Z \big]
	= \E\big[ l(\eta Y_T + (1-\eta)Y_T^0) Z_k^\eta \big].
\end{equation}
Moreover, one has 
$\| Z_k^\eta \|_p \leq_c \frac{\| Z \|_{k,2p}}{((1-\rho^2) \lambda_{\mathrm{inf}} \sqrt{T})^k}$, uniformly in $\eta$.
\end{lemma}

\begin{proof}
We prove the lemma for $k=1$, since for $k=2$ the proof is similar. 

\begin{itemize}[leftmargin=*]
\item[] \textit{Step 1:} $F_{\eta}=\eta Y_T + (1-\eta)Y_T^0$ \textit{is a non-degenerate random variable (in the Malliavin sense).}

Under $(\mathbb R_3)$, we know that $F_{\eta}$ is in $\cap_{p \geq 1} \D^{3,p}$. 
One has to prove that the Malliavin covariance matrix associated to $F_{\eta}$, which is a scalar in this case, is defined as
\begin{equation}\label{gammaFeta}
\gamma_{F_{\eta}} 
	= \int_0^T (D^{1}_rF_{\eta})^2\ud r + \int_0^T (D^{2}_rF_{\eta})^2 \ud r
\end{equation}
is almost surely positive and its inverse is in any ${L}^p (p \geq 1)$.

By linearity, we have 
\begin{equation}
D^1_r F_{\eta} = \eta  D^1_r Y_T + (1-\eta)D^1_rY_T^0.
\end{equation}
From \eqref{YSDEMalliavinFirstDeriv} and by setting $\eta$ to $1$ and $0$ successively, we get for $r \leq t$
\begin{equation}\label{YSDEMalliavinFirstDerivEta1Eta0}
\left\lbrace
\begin{array}{lll}
D^1_r Y_t 
	&=&  \lambda(r, Y_r) \sqrt{1-\rho^2}  + \int_r^t D^1_r Y_u( \alpha_y \ud u + \lambda_y (\sqrt{1-\rho^2} \ud \widetilde{W}_u^L + \rho \ud W_u^X)),\\
D^1_r Y^0_t 
	&=& \lambda(r, y_0)\sqrt{1-\rho^2}.
\end{array} \right.
\end{equation}
By solving \eqref{YSDEMalliavinFirstDerivEta1Eta0} for $r \leq T$, we obtain
 \begin{equation}\label{YSDEMalliavinFirstDerivEta1Eta0Solution}
\left\lbrace
\begin{array}{lll}
D^1_r Y_T 
	&=&  \lambda(r, Y_r) \sqrt{1-\rho^2} \e^{ \int_r^T \alpha_y \ud u -\frac{1}{2}\int_r^T \lambda_y^2 \ud u  
		+  \int_r^T \lambda_y (\sqrt{1-\rho^2} \ud \widetilde{W}_u^L + \rho \ud W_u^X)} \\
D^1_r Y_T^0 &=& \lambda(r, y_0)\sqrt{1-\rho^2}.
\end{array} \right.
\end{equation}
Hence, we can write
{\small{
\begin{equation}\label{FinalizeProofofNonDegenerate}
\begin{array}{lll}
\gamma_{F_{\eta}} &\geq&  \int_0^T (D^{1}_rF_{\eta})^2 \ud r \\
				  &\geq& \int_0^T \left[ \lambda(r, Y_r) 
\sqrt{1-\rho^2} \e^{ \int_r^T \alpha_y \ud u -\frac{1}{2}\int_r^T \lambda_y^2 \ud u  + 
 \int_r^T \lambda_y (\sqrt{1-\rho^2} 
\ud \widetilde{W}_u^L + \rho \ud W_u^X)} + (1-\eta) \lambda(r, y_0)\sqrt{1-\rho^2}  
\right]^2 \ud r \\
				   &\geq& (\inf_{0 \leq r \leq T} \e^{ \int_r^T 
\alpha_y \ud u -\frac{1}{2}\int_r^T \lambda_y^2 \ud u  +  
\int_r^T \lambda_y (\sqrt{1-\rho^2} \ud \widetilde{W}_u^L + \rho \ud W_u^X)})^2 (1-\rho^2) \int_0^T ( \lambda(r, Y_r) + 
(1-\eta) \lambda(r, y_0) )^2 \ud r\\
				   &\geq& (\inf_{0 \leq r \leq T} \e^{ \int_r^T \alpha_y \ud u -\frac{1}{2}\int_r^T \lambda_y^2 \ud u  +  
\int_r^T \lambda_y (\sqrt{1-\rho^2} \ud \widetilde{W}_u^L + \rho \ud W_u^X)})^2 T \lambda_{\mathrm{inf}}^2 (1-\rho^2).\\ 
\end{array} 
\end{equation} }}
The second inequality shows that $\gamma_{F_{\eta}}$ is almost surely positive. 
With the last inequality and the control of the moments for the solution of an SDE (see \citet[Section 6.2.1]{Priouret05}), we get for $p \geq 1$
\begin{equation}\label{momentestimatesgammaFeta}
\| \gamma_{F_{\eta}}^{-1} \|_p \leq_c \frac{1}{( \lambda_{\mathrm{inf}}  
\sqrt{1-\rho^2} \sqrt{T})^2}.
\end{equation}

\item[] \textit{Step 2:} \textit{Integration-by-parts formula}.

Using Propositions 2.1.4 and 1.5.6 in \citet{Nualart05}, one gets the existence of $Z_1^\eta$ in $L^p (p \geq 1)$ with
\begin{equation}
\| Z_1^\eta \|_p \leq_c \| \gamma_{F_{\eta}}^{-1} \|_{1,4p} 
\| DF_{\eta} \|_{1,4p} \| Z  \|_{1,2p}.
\end{equation}

\item[] \textit{Step 3:} \textit{Upper bound for $\| DF_{\eta} \|_{1,q}$,  $\| \gamma_{F_{\eta}}^{-1} \|_{1,q}$, for $q \geq 2$}.

We recall that
\begin{equation}
\| DF_{\eta} \|_{1,q}^q 
	= \E\left[  \left( \sum_{i=1}^2 \int_0^T (D^i_t F_{\eta})^2 dt \right)^{\frac{q}{2}}  \right]  
	+ \E\left[  \left( \sum_{i,j=1}^2 \int_0^T (D^{i,j}_{t_i, t_j} F_{\eta})^2 dt_idt_j \right)^{\frac{q}{2}}  \right].
\end{equation}
We bound each term above separately using the linearity of the Malliavin derivative operator to $F_{\eta}$, the Holder inequality and the Malliavin derivatives' estimates in Lemma \ref{EstimatesMalliavinDerivatives}, to obtain 
\begin{equation}\label{inequalityDF}
\| DF_{\eta} \|_{1,q} \leq_c \sqrt{T} |\lambda|_{\infty}.
\end{equation} 
For $\| \gamma_{F_{\eta}}^{-1} \|_{1,q}$ which is given by
\begin{equation}
\| \gamma_{F_{\eta}}^{-1} \|_{1,q} 
	= \E |\gamma_{F_{\eta}}^{-1}|^q  + \E\left[  \left( \sum_{i=1}^2 \int_0^T (D^i_t \gamma_{F_{\eta}}^{-1} )^2 dt \right)^{\frac{q}{2}}  \right],
\end{equation}
we use Lemma 2.1.6 in \cite{Nualart05} to write $D^i_t \gamma_{F_{\eta}}^{-1} = -\frac{D^i_t \gamma_{F_{\eta}}}{\gamma^2_{F_{\eta}}}, \, i=1,2$.
 
Similarly, we bound each term above separately using the linearity of the Malliavin derivative operator to $\gamma_{F_{\eta}}$ (see \eqref{gammaFeta}), the H\"older inequality, the Malliavin derivatives' estimates in Lemma \ref{EstimatesMalliavinDerivatives} and the moments estimates in \eqref{momentestimatesgammaFeta} to obtain 
\begin{equation}\label{inequalitygammaminus1}
\| \gamma_{F_{\eta}}^{-1} \|_{1,q} \leq_c \frac{1}{( \lambda_{\mathrm{inf}}  \sqrt{1-\rho^2} \sqrt{T})^2}.
\end{equation}
Finally, using $|\lambda|_{\infty} \leq C_E \lambda_{\mathrm{inf}}$ (Assumption ($\mathbb{ELL}$)) combined with inequalities \eqref{inequalityDF} and \eqref{inequalitygammaminus1}, we get  
\begin{equation}
\| \gamma_{F_{\eta}}^{-1} \|_{1,4p} 
\| DF_{\eta} \|_{1,4p}\leq_c \frac{1}{(1-\rho^2) \lambda_{\mathrm{inf}} 
\sqrt{T}}.
\end{equation}
This completes our proof. \qedhere
\end{itemize}
\end{proof}

\section{Derivation of the third order approximation formula} \label{DerivThirdOrderFormula}\label{subsubsection:GreeksCoefThirdOrder}

The additional correction terms for the third order expansion formula in \eqref{ThirdOrderExpPayoffanderror} are provided by $\E\big[h^{(1)}(Y_T^0)\frac{Y_{2,T}}{2}\big]$ and $\E\big[h^{(2)}(Y_T^0) \frac{(Y_{1,T})^2}{2}\big]$. 
The first term, setting $\eta=0$ in \eqref{SDESecondDerivative}, yields
\begin{align}\label{3rdorder1stterm}
\E\left[h^{(1)}(Y_T^0)\frac{Y_{2,T}}{2}\right]
	&= \E\left[h^{(1)}(Y_T^0) \int_0^T Y_{1,t}\alpha_y \dt\right]
	 + \E\left[h^{(1)}(Y_T^0) \int_0^T Y_{1,t}\lambda_y\dW^L_t\right] \nonumber \\ \nonumber
	&\quad+ \E\left[h^{(1)}(Y_T^0) \int_0^TZ_{1,t}\alpha_z\dt\right]
	 + \E\left[h^{(1)}(Y_T^0) \int_0^T \xi_t\gamma_t \alpha_{yz}\dt\right] \\ \nonumber 
	&\quad+ \E\left[h^{(1)}(Y_T^0) \int_0^T\frac{\xi_t^2}{2}\alpha_{yy}\dt\right]
	 + \E\left[h^{(1)}(Y_T^0) \int_0^T \frac{\gamma^2_t}{2}\alpha_{zz}\dt\right] \\
 	&\quad+ \E\left[h^{(1)}(Y_T^0) \int_0^T \frac{\xi^2_t}{2}\lambda_{yy}\dW^L_t\right].
\end{align}
Using Lemma \ref{lem3ndOrderExpansionFormula}, we compute each of the sub-correction terms separately:
\begin{multline}
\E\left[ h^{(1)}(Y_T^0) \left(\int_0^T Y_{1,t}\alpha_y \dt \right)\right] 
	= [w(\alpha, \alpha_y, \alpha_y)_0^T + w(\beta,\alpha_z, \alpha_y)_0^T] g^h_1(Y_T^0) \\
	+ [w(\lambda^2, \alpha_y, \alpha_y)_0^T 
	  	+ w(\alpha, \lambda_y\lambda,\alpha_y)_0^T] g_2^h(Y_T^0)
  	+ w(\lambda^2, \lambda_y\lambda, \alpha_y)  g_3^h(Y_T^0)\\
  	+ \rho w(\sigma \lambda, \alpha_z, \alpha_y)_0^T g_2^h(Y_T^0),
\end{multline}
\begin{multline}
\E\left[ h^{(1)}(Y_T^0) \left(\int_0^T Y_{1,t}\lambda_y \dW_t^L \right)\right]
	= [w(\alpha, \alpha_y, \lambda_y\lambda)_0^T + w(\beta, \alpha_z, \lambda_y\lambda)_0^T] g^h_2(Y_T^0)\\
	+ [w(\lambda^2, \alpha_y, \lambda_y\lambda)_0^T  
		+ w(\alpha,\lambda_y\lambda, \lambda_y\lambda)_0^T] g^h_3(Y_T^0)
	+ w(\lambda^2,\lambda_y\lambda, \lambda_y\lambda)_0^T g^h_4(Y_T^0)\\
	+ \rho w(\sigma\lambda,\alpha_z,\lambda_y\lambda)_0^T g^h_3(Y_T^0),
\end{multline}
\begin{multline}
\E\left[ h^{(1)}(Y_T^0) \left(\int_0^T Z_{1,t}\alpha_z \dt \right)\right] 
	= w(\beta, \beta_z,\alpha_z)_0^T g^h_1(Y_T^0) \\
	+ \rho g^h_2(Y_T^0) \left[ w(\sigma \lambda, \beta_z,\alpha_z)_0^T 
		+ w(\beta, \sigma_z \lambda,\alpha_z)_0^T \right] 
 	+ \rho^2 w(\sigma \lambda, \sigma_z \lambda,\alpha_z)_0^T g^h_3(Y_T^0),
\end{multline}
\begin{multline}
\E\left[ h^{(1)}(Y_T^0) \left(\int_0^T\xi_t \gamma_t \alpha_{yz} \dt \right)\right]
	= \left[ \omega(\alpha, \beta ,\alpha_{yz})_0^T 
		+ \omega(\beta, \alpha, \alpha_{yz})_0^T \right] g_1^h(Y_T^0) \\
	+ \left[ \omega(\lambda^2, \beta,\alpha_{yz})_0^T 
		+ \omega(\beta, \lambda^2,\alpha_{yz})_0^T\right] g_2^h(Y_T^0) \\
	+  \rho \Big[ \omega(\lambda \sigma, \alpha_{yz})_0^T g_1^h(Y_T^0) 
		+ \left( \omega(\sigma \lambda, \alpha,\alpha_{yz})_0^T 
			+ \omega( \alpha, \sigma \lambda,\alpha_{yz})_0^T  \right) g_2^h(Y_T^0) \\ 
		+ \left( \omega( \lambda^2, \sigma \lambda,\alpha_{yz})_0^T 
			+ \omega( \sigma \lambda, \lambda^2,\alpha_{yz})_0^T  \right) g_3^h(Y_T^0)\Big],
\end{multline}
\begin{multline}
\E\left[ h^{(1)}(Y_T^0) \left(\int_0^T\frac{\xi_t^2}{2} \alpha_{yy} dt\right) \right]
	= [w(\alpha, \alpha, \alpha_{yy})_0^T	
		+ \frac{1}{2}w( {\lambda^2}, \alpha_{yy})_0^T] g^h_1(Y_T^0)\\
	+ [w(\alpha, \lambda^2, \alpha_{yy})_0^T 
		+ w(\lambda^2,\alpha, \alpha_{yy})_0^T] g^h_2(Y_T^0)
	+ w(\lambda^2,\lambda^2, \alpha_{yy})_0^Tg^h_3(Y_T^0),
\end{multline}
\begin{multline}
\E\left[ h^{(1)}(Y_T^0) \left(\int_0^T\frac{\gamma_t^2}{2} \alpha_{zz} \dt\right)\right]
	= \big[ \frac{1}{2}w(\sigma^2, \alpha_{zz})_0^T 
		+ w(\beta,\beta, \alpha_{zz} )_0^T \big] g^h_1(Y_T^0) \\
	+ \rho \left[ w(\sigma\lambda, \beta ,\alpha_{zz} )_0^T 
	+ w(\beta,\sigma\lambda, \alpha_{zz} )_0^T\right] g^h_2(Y_T^0)
	+ \rho^2w(\sigma\lambda,\sigma\lambda, \alpha_{zz} )_0^T g^h_3(Y_T^0),
\end{multline}
\begin{multline}
\E\left[ h^{(1)}(Y_T^0) \left(\int_0^T\frac{\xi_t^2}{2} \lambda_{yy} \ud W_t^L\right)\right]
	= \E\left[h^{(2)}(Y_T^0) \left(\int_0^T\frac{\xi_t^2}{2} \lambda_{yy}\lambda_t \dt \right)\right] = \\
	= \big[ \frac{1}{2}w(\lambda^2, \lambda\lambda_{yy})_0^T 
	+ w(\alpha, \alpha,\lambda\lambda_{yy})_0^T \big]g^h_2(Y_T^0) \\
	+ \left[ w(\lambda^2, \alpha,\lambda\lambda_{yy})_0^T 
	+ w(\alpha, \lambda^2,\lambda\lambda_{yy})_0^T \right]g^h_3(Y_T^0) 
	+ w(\lambda^2, \lambda^2,\lambda\lambda_{yy})_0^T g^h_4(Y_T^0),
\end{multline}

As for the second corrective term, by applying It\^o's formula to ($Y^2_{1,t})_{t \geq 0}$, we obtain
\begin{eqnarray}\label{3rdorderscdterm}\nonumber 
\E\left[\frac{h^{(2)}(Y_T^0)}{2}Y^2_{1,T}\right]
	&=& \E\left[h^{(2)}(Y_T^0)\int_0^T Y_{1,t} \xi_t \alpha_y \dt\right]
	+   \E\left[h^{(2)}(Y_T^0)\int_0^T Y_{1,t}\gamma_t\alpha_{z}\dt\right] \\\nonumber 
	&&+ \E\left[\frac{h^{(2)}(Y_T^0)}{2}\int_0^T\xi_t^2 \lambda_y^2\dt\right]
	+	\E\left[h^{(2)}(Y_T^0)\int_0^T Y_{1,t}\xi_t \lambda_y \ud W_t^L\right] \\
 	&=:& A + B + C + D.
\end{eqnarray}
We get directionsly for $C$, by applying \eqref{equality5} with $\theta = \lambda_y^2$, that
\begin{align}
C &= \big[ \frac12 \omega(\lambda^2, \lambda_y^2)_0^T 
		+ \omega(\alpha, \alpha, \lambda_y^2)_0^T \big] g_2^h(Y_T^0) 
	+  \left[ \omega(\lambda^2, \alpha,\lambda_y^2)_0^T 
		+ \omega(\alpha,\lambda^2,\lambda_y^2)_0^T\right] g_3^h(Y_T^0)\nonumber \\
  &\quad+ \omega(\lambda^2, \lambda^2,\lambda_y^2)_0^T g_4^h(Y_T^0).
\end{align}
Moreover, by applying Lemma \ref{IBP2} to $D$, we get
\begin{equation}\label{DFormula}
D = \E\left[h^{(3)}(Y_T^0)\int_0^T Y_{1,t}\xi_t \lambda_y \lambda \dt\right],
\end{equation}
hence $A$ and $D$ can be computed similarly. 
Indeed, with It\^o's formula and successive applications of Lemmata \ref{IBP1} and \ref{IBP2}, we obtain for $D$
\begin{align}\label{Dformula1}
D  &= \E\left[h^{(3)}(Y_T^0)\int_0^T Y_{1,t} \alpha \omega(\lambda \lambda_y)_t^T \dt\right] 
	+ \E\left[h^{(3)}(Y_T^0)\int_0^T \xi_t^2\alpha_y \omega(\lambda \lambda_y)_t^T \dt\right] \nonumber \\ 
   &\quad+ \E\left[h^{(3)}(Y_T^0)\int_0^T \xi_t \gamma_t\alpha_z \omega(\lambda \lambda_y)_t^T \dt\right]
	+ \E\left[h^{(3)}(Y_T^0)\int_0^T \xi_t \lambda \lambda_y \omega(\lambda \lambda_y)_t^T \dt\right] \nonumber \\
   &\quad+ \E\left[h^{(4)}(Y_T^0)\int_0^T Y_{1,t} \lambda^2 \omega(\lambda \lambda_y)_t^T \dt\right]
  	+ \E\left[h^{(4)}(Y_T^0)\int_0^T \xi_t^2 \lambda \lambda_y \omega(\lambda \lambda_y)_t^T \dt\right] \nonumber \\
   &=: D_1 +  D_2 + D_3 + D_4 + D_5 + D_6.
\end{align}
Each term $D_i$ is computed explicitly using Lemma \ref{lem3ndOrderExpansionFormula}.
Furthermore, with 
\begin{equation}
	\alpha_y(t,y,z) = -[\lambda_y(t,y)\lambda(t,y) + \rho \lambda_y(t,y) \sigma(t,z)]	
\end{equation}
we deduce directly the following expression for $A$:
\begin{align}\label{ExpressionforA}
A 	&= \E\left[h^{(2)}(Y_T^0)\int_0^T Y_{1,t} \xi_t \alpha_y \dt\right] \nonumber \\
	&= -\E\left[h^{(2)}(Y_T^0)\int_0^T Y_{1,t} \xi_t \lambda_y \lambda \dt\right]
	   -\rho \E\left[h^{(2)}(Y_T^0)\int_0^T Y_{1,t} \xi_t \lambda_y \sigma \dt\right].
\end{align}
Each term above can be deduced from $D$. 
Finally, the expression for $B$ is given by
\begin{align}
B 	&= \E\left[h^{(2)}(Y_T^0)\int_0^T Y_{1,t}\beta \omega(\alpha_z)_t^T \dt\right] 
	 + \E\left[h^{(2)}(Y_T^0)\int_0^T \gamma_t^2\alpha_z \omega(\alpha_z)_t^T \dt\right] \nonumber\\
	&\quad+ \E\left[h^{(2)}(Y_T^0)\int_0^T \xi_t \gamma_t \alpha_y \omega(\alpha_z)_t^T \dt\right]
	+ \rho \E\left[h^{(2)}(Y_T^0)\int_0^T \xi_t \sigma \lambda_y \omega(\alpha_z)_t^T \dt\right] \nonumber\\
	&\quad+ \rho \E\left[h^{(3)}(Y_T^0)\int_0^T Y_{1,t} \lambda \sigma \omega(\alpha_z)_t^T \dt\right]
	+ \E\left[h^{(3)}(Y_T^0)\int_0^T \xi_t \gamma_t \lambda \lambda_y \omega(\alpha_z)_t^T \dt\right] \nonumber\\
	&=: B_1 +  B_2 + B_3 + B_4 + B_5 + B_6.
\end{align}
By gathering all these terms, passing to the initial parameters and writing them as a polynomial function of $\rho$ (order 4), we obtain the third order expansion formulas in \eqref{3rdOrderExpansionFormula}. 
We omit the details of these computations for the sake of brevity and provide directly the results. 
The constant coefficient is as in Theorem \ref{ApproxFormulaProbaApproa3order}. 
The other coefficients are provided below.  

We have that 
\begin{align}
\gamma_{1,T}  
	&= \sum_{j=1}^5 \gamma_{1,j,T}g_j^h(Y_T^0) \label{gamma1} \\
\gamma_{1,1,T} 
	&= A_{7,T} +\frac{1}{2}(A_{8,T} + A_{9,T} -A_{10,T}-A_{11,T}) -B_{28,T} \nonumber \\
	& -\frac{1}{2}(B_{26,T}+ B_{27,T} +B_{32,T} + B_{36,T} + B_{35,T} + B_{31,T} +B_{42,T}) \nonumber\\
	& -\frac{1}{4}(B_{33,T} +  B_{34,T} + B_{29,T}+B_{37,T}) \\
\gamma_{1,2,T} 
	&= -A_{7,T}-A_{8,T} + B_{29,T} + \frac{5}{2}B_{27,T} + 3(B_{28,T} + B_{26,T}) +\frac{1}{2}(B_{32,T} + B_{33,T} + B_{34,T}) \nonumber\\
	& + \frac{3}{2}(B_{31,T} + B_{30,T} + B_{42,T} + B_{35,T}) +  C_{56,T} + C_{55,T} +\frac{1}{2} (C_{50,T} + C_{52,T} + C_{54,T} \nonumber\\
  	&\quad + C_{78,T}) + \frac{1}{4}(C_{5,T} + C_{6,T} + C_{7,T}  + C_{51,T} + C_{53,T} + C_{57,T} + C_{58,T} + C_{77,T})  \\
\gamma_{1,3,T} 
  	&= -3B_{26,T} -2(B_{27,T}+B_{28,T}) -(B_{29,T}+B_{30,T}+B_{31,T}+B_{35,T}+B_{42,T})\nonumber \\ 
  	& -\frac{3}{4}( C_{5,T} + C_{6,T} + C_{7,T} +C_{53,T} + C_{57,T} + C_{58,T}) -4(C_{55,T} +C_{56,T})  \nonumber \\
  	& -2(C_{50,T} + C_{52,T}) -\frac{5}{2} (C_{54,T} + C_{78,T}) -\frac{5}{4}(C_{51,T} + C_{77,T})  \\
\gamma_{1,4,T} 
  	&= \frac{1}{2}( C_{5,T} + C_{6,T} + C_{7,T} +C_{57,T} + C_{58,T} + C_{53,T} )+ 2(C_{51,T}+C_{77,T}) \nonumber \\
    &  +4 (C_{54,T}+C_{78,T}) + 5(C_{55,T} + C_{56,T}) + \frac{5}{2}(C_{50,T}+ C_{52,T}) \\ 
\gamma_{1,5,T} 
	&= -C_{50,T} - C_{51,T} - C_{52,T}  -C_{77,T} -2( C_{54,T} + C_{55,T} + C_{56,T} + C_{78,T}),
\end{align} 
and 
\begin{align}
\gamma_{2,T}  
	&= \sum_{j=1}^5 \gamma_{2,j,T}g_j^h(Y_T^0)   \label{gamma2}\\
\gamma_{2,1,T} 
	&= A_{4,T} - A_{5,T} - \frac{1}{2}( B_{10,T}	+ B_{13,T} + B_{14,T} + B_{15,T} + B_{16,T} + B_{17,T} ) \nonumber \\ 
	&	- B_{11,T} - B_{12,T}  - B_{18,T} - B_{19,T}   \\
\gamma_{2,2,T} 
	&= -A_{6,T} + 2(B_{10,T}+B_{20,T}+B_{11,T}) + B_{12,T} +  B_{41,T} + B_{16,T}+B_{17,T} + B_{18,T} \nonumber \\
	&  +B_{19,T}  + B_{38,T} + \frac{3}{2}B_{21,T} +  \frac{1}{2}(B_{22,T} + B_{23,T} +B_{24,T} + B_{25,T}+ B_{39,T}) \nonumber \\
	& + C_{39,T} + C_{42,T} + C_{43,T} + C_{84,T} + C_{85,T} \nonumber \\
 	& + 2C_{44,T} +\frac{1}{4} (C_{8,T} + C_{10,T} + C_{12,T} + C_{14,T}+C_{80,T} + C_{82,T}+C_{86,T}+C_{87,T}) \nonumber  \\
	& + \frac{1}{2} (C_{9,T} + C_{11,T} + C_{15,T}  + C_{16,T} + C_{38,T} + C_{40,T} + C_{41,T} + C_{45,T} + C_{46,T} \nonumber \\
	& + C_{79,T} + C_{81,T} + C_{83,T})  
\end{align}
\begin{align}	
\gamma_{2,3,T} 
 	&= -(B_{23,T}+B_{22,T} + B_{38,T} + B_{39,T}) -2B_{20,T} \nonumber \\
	&   -C_{80,T} -2(C_{83,T}+C_{39,T}) -3(C_{42,T}+C_{43,T} +C_{84,T} + C_{85,T}) -4C_{44,T} \nonumber \\
 	& -\frac{1}{2}( C_{8,T} + C_{9,T} + C_{17,T} + C_{15,T} + C_{16,T} + C_{12,T} + C_{14,T} + C_{18,T} + C_{19,T})\nonumber \\
	& -\frac{3}{2}(C_{38,T} + C_{40,T} + C_{79,T} + C_{81,T}) \nonumber \\
	& -\frac{1}{2}(  C_{41,T} + C_{45,T} + C_{46,T} +C_{47,T} + C_{48,T} + C_{49,T} + C_{86,T} + C_{87,T} + C_{82,T}) \\
\gamma_{2,4,T} 
	&= C_{38,T} + C_{39,T} + C_{40,T} + C_{79,T} + C_{80,T} + C_{81,T} \nonumber \\ 
	& +2(C_{44,T} + C_{42,T} + C_{43,T} + C_{85,T} + C_{83,T} + C_{84,T})\nonumber  \\
	& + \frac{3}{2}( C_{17,T} + C_{19,T} + C_{18,T} +C_{47,T}+C_{48,T}+C_{49,T}) \\
\gamma_{2,5,T} 
	&= -(C_{17,T} + C_{19,T} + C_{18,T} + C_{49,T} + C_{48,T} + C_{47,T}), 
\end{align}
and
\begin{align}
\gamma_{3,T}  
	&= \sum_{j=1}^4 \gamma_{3,j,T}g_j^h(Y_T^0)   \label{gamma3}\\
\gamma_{3,1,T} 
	&=  -B_{4,T} -B_{8,T}  \\ 
\gamma_{3,2,T} 
	&= 2(B_{5,T}+B_{7,T}+B_{40,T}) + C_{64,T} + C_{67,T} + C_{68,T} +  C_{34,T}+ 2(C_{69,T} + C_{35,T}) \nonumber \\
	&+ \frac{1}{2}(C_{20,T} +  C_{21,T} + C_{22,T} + C_{63,T} +C_{65,T} +  C_{66,T} + C_{70,T} + C_{71,T}) \\
\gamma_{3,3,T} 
	&= -B_{6,T} -B_{9,T} - ( C_{31,T} + C_{24,T} + C_{25,T} + C_{34,T} + C_{63,T} + C_{64,T} + C_{65,T} + C_{36,T}) \nonumber\\
	&  -2( C_{27,T}+C_{35,T}+C_{69,T} + C_{67,T} + C_{68,T} + C_{37,T} ) \nonumber\\
	& -  \frac{1}{2} (C_{26,T} + C_{23,T} + C_{28,T} + C_{29,T} + C_{30,T} + C_{73,T} + C_{74,T} + C_{75,T}) \\
\gamma_{3,4,T} 
	&= C_{28,T} + C_{26,T} +  C_{30,T} + C_{31,T} + C_{75,T}  + C_{74,T} +  C_{73,T} +C_{36,T} \nonumber\\
	&+ 2 (C_{27,T}+ C_{37,T}),
\end{align}
and finally
\begin{align}
\gamma_{4,T}  	&= \sum_{j=2}^4 \gamma_{4,j,T}g_j^h(Y_T^0)   \label{gamma4}\\
\gamma_{4,2,T}	&= C_{59,T}+2C_{60,T} \\
\gamma_{4,3,T} 	&= - C_{4,T} -C_{61,T} -2(C_{62,T}+C_{3,T}) \\
\gamma_{4,4,T} 	&= 2C_{1,T} + C_{2,T}  	
\end{align}
All the expressions for the coefficients $A_{i,T}$, $B_{i,T}$ and $C_{i,T}$ are gathered in Tables \ref{TabforAcoeff}, \ref{TabforBcoeff} and \ref{TabforCcoeff} below. 

\begin{table}[htbp]
  \centering
\begin{tabular}{|c|c|c|c|}
\hline 
$A_{1,T} = \omega(\lambda^2,\lambda\lambda_y)_0^T$ &  
$A_{2,T}=\omega(\lambda^2,\lambda_{yy}\lambda)_0^T $ & 
$A_{3,T}=\omega(\lambda^2,\lambda_y^2)_0^T$ & 
$A_{4,T}=\omega(\lambda \sigma,\lambda_y \sigma)_0^T$\\
$A_{5,T}=\omega(\lambda \sigma,\lambda_y \sigma_z)_0^T$ & 
$A_{6,T}=\omega(\lambda \sigma,\lambda \sigma_z)_0^T$ &  
$A_{7,T}=\omega(\lambda \sigma,\lambda_y \lambda)_0^T$ & 
$A_{8,T}=\omega(\lambda^2,\lambda_y \sigma)_0^T$ \\
$A_{9,T}=\omega( \sigma^2,\lambda \sigma_z)_0^T$ & 
$A_{10,T}=\omega(\lambda^2,\lambda_{yy} \sigma)_0^T$ & 
$A_{11,T}=\omega(\sigma^2,\lambda \sigma_{zz})_0^T$& \\
\hline
\end{tabular}\\[.35em]
  \caption{Weight coefficients involving 2 multiple integrals } \label{TabforAcoeff}
\end{table}

\begin{table}[htbp]
  \centering
\begin{tabular}{|c|c|c|}
\hline 
 $B_{1,T}=\omega(\lambda^2,\lambda\lambda_y, \lambda\lambda_y )_0^T$ & $B_{2,T}=\omega(\lambda^2, \lambda^2,\lambda_{yy}\lambda)_0^T$  &
 $B_{3,T}=\omega(\lambda^2, \lambda^2,\lambda_{y}^2)_0^T$ \\
 $B_{4,T}=\omega(\lambda \sigma,\lambda_y \sigma, \lambda_y \sigma )_0^T$ & $B_{5,T}= \omega(\lambda \sigma,\lambda \sigma_z, \lambda_y \sigma )_0^T$ &
 $B_{6,T}= \omega(\lambda \sigma,\lambda \sigma_z, \lambda \sigma_z )_0^T$ \\
 $B_{7,T}=\omega(\lambda \sigma,\lambda \sigma, \lambda_y \sigma_z )_0^T$ & $B_{8,T}= \omega(\lambda \sigma,\lambda \sigma, \lambda_{yy} \sigma )_0^T$ & 
 $B_{9,T}= \omega(\lambda \sigma,\lambda \sigma, \lambda \sigma_{zz} )_0^T$ \\
 $B_{10,T}= \omega(\lambda^2,\lambda_y \sigma, \lambda_y \sigma )_0^T $ & $B_{11,T}= \omega(\lambda \sigma,\lambda_y \lambda, \lambda_y \sigma )_0^T $ & 
 $B_{12,T}= \omega(\lambda \sigma,\lambda_y \sigma, \lambda_y \lambda )_0^T$ \\
 $B_{13,T}= \omega(\sigma^2,\lambda \sigma_z, \lambda_y \sigma )_0^T$ & $B_{14,T}= \omega(\lambda \sigma, \sigma^2, \lambda_y \sigma_z )_0^T$ & 
 $B_{15,T}=  \omega(\sigma^2, \lambda \sigma, \lambda_y \sigma_z )_0^T $\\
 $B_{16,T}= \omega(\lambda^2,\lambda \sigma, \lambda_{yy} \sigma )_0^T$ & $B_{17,T}= \omega(\lambda \sigma, \lambda^2,\lambda_{yy} \sigma )_0^T$ &  
 $B_{18,T}= \omega(\lambda \sigma, \lambda \sigma,\lambda_{yy} \lambda )_0^T$\\
 $B_{19,T}=\omega(\lambda \sigma, \lambda \sigma, \lambda_y^2 )_0^T $ & $B_{20,T}=  \omega(\lambda \sigma, \lambda \sigma_z, \lambda_y \lambda )_0^T$ &  
 $B_{21,T}= \omega(\sigma^2, \lambda \sigma_z, \lambda \sigma_z )_0^T$ \\
 $B_{22,T}=\omega(\sigma \lambda, \lambda^2, \lambda_y \sigma_z )_0^T $ & $B_{23,T}=  \omega(\lambda^2, \sigma \lambda, \lambda_y \sigma_z )_0^T $ & 
 $B_{24,T}= \omega(\lambda \sigma, \sigma^2, \lambda \sigma_{zz} )_0^T$ \\
 $B_{25,T}= \omega(\sigma^2, \lambda \sigma, \lambda \sigma_{zz} )_0^T $& $B_{26,T}= \omega(\lambda^2,\lambda_y \lambda, \lambda_y \sigma )_0^T $& 
 $B_{27,T}= \omega(\lambda^2, \lambda_y \sigma, \lambda_y \lambda )_0^T$\\
 $B_{28,T}=\omega(\lambda \sigma, \lambda_y \lambda, \lambda_y \lambda )_0^T $ & $B_{29,T}= \omega(\lambda^2, \lambda^2,\lambda_{yy}\sigma)_0^T$ & 
 $B_{30,T}= \omega(\lambda^2, \lambda \sigma,\lambda \lambda_{yy})_0^T$\\
 $B_{31,T}= \omega(\lambda \sigma, \lambda^2,\lambda \lambda_{yy})_0^T $ & $B_{32,T}= \omega(\sigma^2, \lambda \sigma_z,\lambda \lambda_{y})_0^T$ & 
 $B_{33,T}= \omega(\lambda^2, \sigma^2, \lambda_y \sigma_z )_0^T $\\
 $B_{34,T}= \omega(\sigma^2, \lambda^2, \lambda_y \sigma_z )_0^T $ &  $B_{35,T}= \omega(\lambda^2, \lambda \sigma, \lambda_y^2 )_0^T$ & 
 $B_{36,T}= \omega(\sigma^2, \sigma \sigma_z, \lambda \sigma_z)_0^T $ \\
 $B_{37,T} = \omega(\sigma^2, \sigma^2, \lambda \sigma_{zz} )_0^T$ &  
 $B_{38,T}= \omega(\lambda \sigma,\lambda \lambda_y , \lambda \sigma_z )_0^T$ &
 $B_{39,T}= \omega(\lambda^2, \sigma \lambda_y, \lambda \sigma_{z} )_0^T $\\ 
 $B_{40,T}=\omega(\lambda \sigma, \sigma \lambda_y, \lambda \sigma_{z} )_0^T $ &
 $B_{41,T}=\omega(\lambda \sigma, \sigma \sigma_z, \lambda \sigma_{z} )_0^T$&
 $B_{42,T}=  \omega(\lambda \sigma, \lambda^2, \lambda_y^2 )_0^T $ \\
\hline
\end{tabular}\\[.35em]
  \caption{Weight coefficients involving 3 multiple integrals } \label{TabforBcoeff}
\end{table}

\begin{table}[htbp]
  \centering
\begin{tabular}{|c|c|c|}
\hline 
 $C_{1,T}= \omega(\lambda \sigma, \lambda \sigma , \lambda \sigma_{z}, \lambda \sigma_{z}  )_0^T $ & 
 $C_{2,T}= \omega(\lambda \sigma, \lambda \sigma_{z},\lambda \sigma, \lambda \sigma_{z}  )_0^T$ &  
 $C_{3,T}= \omega(\lambda \sigma, \lambda \sigma,\lambda_y \sigma, \lambda \sigma_{z}  )_0^T $\\
 $C_{4,T}= \omega(\lambda \sigma, \lambda_y \sigma, \lambda \sigma, \lambda \sigma_{z}  )_0^T$ & 
 $C_{5,T}= \omega(\lambda^2, \lambda_y \lambda, \sigma^2, \lambda \sigma_{z}  )_0^T$ & 
 $C_{6,T}= \omega(\lambda^2, \sigma^2, \lambda_y \lambda, \lambda \sigma_{z}  )_0^T$ \\
 $C_{7,T}= \omega(\sigma^2, \lambda^2, \lambda_y \lambda, \lambda \sigma_{z}  )_0^T$& 
 $C_{8,T}= \omega(\lambda^2, \lambda_y \sigma, \sigma^2, \lambda \sigma_{z}  )_0^T  $ & 
 $C_{9,T}= \omega(\lambda \sigma, \lambda_y \lambda, \sigma^2, \lambda \sigma_{z}  )_0^T$\\
 $C_{10,T}= \omega(\sigma^2, \lambda \sigma_z, \sigma^2, \lambda \sigma_{z}  )_0^T$& 
 $C_{11,T}= \omega(\sigma^2, \sigma^2, \lambda \sigma_z, \lambda \sigma_{z}  )_0^T$ & 
 $C_{12,T}= \omega(\lambda^2, \sigma^2, \lambda_y \sigma, \lambda \sigma_{z}  )_0^T$\\
 $C_{14,T}= \omega(\sigma^2, \lambda^2, \lambda_y \sigma, \lambda \sigma_{z}  )_0^T $&  
 $C_{15,T}= \omega(\lambda \sigma, \sigma^2, \lambda_y \lambda, \lambda \sigma_{z}  )_0^T$  & 
 $C_{16,T}= \omega(\sigma^2, \lambda \sigma, \lambda_y \lambda, \lambda \sigma_{z}  )_0^T$  \\
 $C_{17,T}= \omega(\lambda^2, \lambda_y \lambda, \lambda \sigma, \lambda \sigma_{z}  )_0^T $ & 
 $C_{18,T}= \omega(\lambda \sigma , \lambda^2, \lambda_y \lambda, \lambda \sigma_{z}  )_0^T$ & 
 $C_{19,T}= \omega(\lambda^2,  \sigma \lambda, \lambda_y \lambda , \lambda \sigma_{z}  )_0^T$ \\
 $C_{20,T}= \omega(\lambda \sigma, \lambda_y \sigma, \sigma^2, \lambda \sigma_{z}  )_0^T $& 
 $C_{21,T}= \omega(\lambda \sigma, \sigma^2, \lambda_y \sigma, \lambda \sigma_{z}  )_0^T $ & 
 $C_{22,T}= \omega(\sigma^2, \lambda \sigma, \lambda_y \sigma, \lambda \sigma_{z})_0^T$\\
 $C_{23,T}= \omega(\sigma \lambda, \lambda \sigma_z, \sigma^2, \lambda \sigma_{z}  )_0^T$ &  
 $C_{24,T}=  \omega(\lambda \sigma, \sigma^2, \lambda \sigma_z, \lambda \sigma_{z}  )_0^T$ & 
 $C_{25,T}= \omega(\sigma^2, \lambda \sigma, \lambda \sigma_z, \lambda \sigma_{z}  )_0^T$\\
 $C_{26,T}= \omega(\lambda \sigma , \lambda^2, \lambda_y \sigma, \lambda \sigma_{z}  )_0^T $ &  
 $C_{27,T}= \omega(\lambda \sigma,  \sigma \lambda, \lambda_y \lambda , \lambda \sigma_{z}  )_0^T$ & 
 $C_{28,T}= \omega(\lambda^2,  \sigma \lambda, \lambda_y \sigma , \lambda \sigma_{z}  )_0^T$\\
 $C_{29,T}= \omega(\sigma^2, \lambda \sigma_z, \lambda \sigma, \lambda \sigma_{z}  )_0^T $ &
 $C_{30,T}= \omega(\lambda^2, \lambda_y \sigma, \lambda \sigma, \lambda \sigma_{z}  )_0^T$ &
 $C_{31,T}= \omega(\lambda \sigma, \lambda_y \lambda, \lambda \sigma, \lambda \sigma_{z}  )_0^T$\\
 $C_{32,T}= \omega(\lambda^2, \lambda_y \lambda, \lambda^2, \lambda \lambda_y   )_0^T$&
 $C_{33,T}= \omega(\lambda^2, \lambda^2, \lambda_y \lambda , \lambda \lambda_y   )_0^T $ &
 $C_{34,T}= \omega(\lambda \sigma, \lambda_y \sigma, \lambda \sigma , \lambda \lambda_y   )_0^T  $\\
 $C_{35,T}= \omega(\lambda \sigma, \lambda \sigma, \lambda_y \sigma, \lambda \lambda_y   )_0^T$ & 
 $C_{36,T}=  \omega(\lambda \sigma, \lambda \sigma_z, \lambda \sigma, \lambda \lambda_y   )_0^T$ & 
 $C_{37,T}= \omega(\lambda \sigma, \lambda \sigma, \lambda \sigma_z, \lambda \lambda_y   )_0^T $\\
 $C_{38,T}= \omega(\lambda^2, \lambda_y \sigma, \lambda \sigma, \lambda \lambda_y   )_0^T $& 
 $C_{39,T}= \omega(\lambda \sigma, \lambda_y \lambda, \lambda \sigma, \lambda \lambda_y   )_0^T $ & 
 $C_{40,T}=\omega(\lambda \sigma, \lambda_y \sigma, \lambda^2, \lambda \lambda_y   )_0^T$ \\
 $C_{41,T}= \omega( \sigma^2, \lambda \sigma_z, \lambda \sigma, \lambda \lambda_y   )_0^T $& 
 $C_{42,T}= \omega(\lambda^2, \lambda \sigma, \lambda_y \sigma, \lambda \lambda_y   )_0^T $& 
 $C_{43,T}= \omega(\lambda \sigma, \lambda^2, \lambda_y \sigma, \lambda \lambda_y   )_0^T$\\
 $C_{44,T}= \omega(\lambda \sigma, \lambda \sigma, \lambda_y \lambda , \lambda \lambda_y   )_0^T$ & 
 $C_{45,T}= \omega(\lambda \sigma, \sigma^2, \lambda \sigma_z, \lambda \lambda_y   )_0^T $& 
 $C_{46,T}=  \omega(\sigma^2, \lambda \sigma, \lambda \sigma_z, \lambda \lambda_y   )_0^T$\\
 $C_{47,T}= \omega(\lambda \sigma, \lambda \sigma_z, \lambda^2, \lambda \lambda_y   )_0^T $ & 
 $C_{48,T}= \omega(\lambda \sigma, \lambda^2, \lambda \sigma_z, \lambda \lambda_y   )_0^T $ & 
 $C_{49,T}= \omega(\lambda^2, \lambda \sigma, \lambda \sigma_z, \lambda \lambda_y   )_0^T $ \\
 $C_{50,T}= \omega(\lambda^2, \lambda_y \lambda , \lambda \sigma, \lambda \lambda_y   )_0^T $ & 
 $C_{51,T}= \omega(\lambda^2, \lambda_y \sigma, \lambda^2, \lambda \lambda_y   )_0^T $ & 
 $C_{52,T}=  \omega(\lambda \sigma , \lambda_y \lambda , \lambda^2, \lambda \lambda_y   )_0^T $\\
 $C_{53,T}=  \omega(\sigma^2, \lambda \sigma_z, \lambda^2, \lambda \lambda_y   )_0^T $ & 
 $C_{54,T}= \omega(\lambda^2, \lambda^2, \lambda_y \sigma , \lambda \lambda_y   )_0^T $ & 
 $C_{55,T}= \omega(\lambda^2, \lambda \sigma, \lambda_y \lambda , \lambda \lambda_y   )_0^T $ \\
 $C_{56,T}= \omega(\lambda \sigma, \lambda^2, \lambda_y \lambda , \lambda \lambda_y   )_0^T $& 
 $C_{57,T}=  \omega(\lambda^2, \sigma^2, \lambda \sigma_z , \lambda \lambda_y   )_0^T $& 
 $C_{58,T}= \omega(\sigma^2, \lambda^2, \lambda \sigma_z , \lambda \lambda_y   )_0^T $\\
 $C_{59,T}= \omega(\lambda \sigma, \lambda_y \sigma, \lambda \sigma ,  \lambda_y \sigma)_0^T $& 
 $C_{60,T}=  \omega(\lambda \sigma, \lambda \sigma , \lambda_y \sigma, \lambda_y \sigma)_0^T $& 
 $C_{61,T}= \omega(\lambda \sigma, \lambda \sigma_z , \lambda \sigma, \lambda_y \sigma)_0^T $ \\
 $C_{62,T}= \omega(\lambda \sigma, \lambda \sigma, \lambda \sigma_z, \lambda_y \sigma)_0^T $& 
 $C_{63,T}=  \omega(\lambda^2 , \lambda_y \sigma, \lambda \sigma, \lambda_y \sigma)_0^T $& 
 $C_{64,T}=  \omega(\lambda \sigma , \lambda_y \lambda, \lambda \sigma, \lambda_y \sigma)_0^T $\\
 $C_{65,T}= \omega(\lambda \sigma , \lambda_y \sigma , \lambda^2 , \lambda_y \sigma)_0^T $& 
 $C_{66,T}= \omega(\sigma^2 , \lambda \sigma_z , \lambda \sigma  , \lambda_y \sigma)_0^T $& 
 $C_{67,T}= \omega(\lambda^2 , \lambda \sigma, \lambda_y \sigma, \lambda_y \sigma)_0^T$\\
 $C_{68,T}= \omega(\lambda \sigma, \lambda^2, \lambda_y \sigma, \lambda_y \sigma)_0^T $& 
 $C_{69,T}= \omega(\lambda \sigma, \lambda \sigma, \lambda_y \lambda, \lambda_y \sigma)_0^T $& 
 $C_{70,T}= \omega(\lambda \sigma, \sigma^2, \lambda \sigma_z, \lambda_y \sigma)_0^T$\\
 $C_{71,T}= \omega(\sigma^2 , \lambda \sigma , \lambda \sigma_z  , \lambda_y \sigma)_0^T$& 
 $C_{72,T}=  \omega(\lambda \sigma, \lambda_{y} \lambda, \lambda \sigma, \lambda_y \sigma)_0^T$ & 
 $C_{73,T}= \omega(\lambda \sigma , \lambda \sigma_z , \lambda^2 , \lambda_y \sigma)_0^T $\\
 $C_{74,T}= \omega(\lambda \sigma, \lambda^2, \lambda \sigma_z, \lambda_y \sigma)_0^T$& 
 $C_{75,T}= \omega(\lambda^2 , \lambda \sigma, \lambda \sigma_z, \lambda_y \sigma)_0^T $& 
 $C_{76,T}= \omega(\lambda \sigma, \lambda \sigma, \lambda \lambda_y , \lambda_y \sigma)_0^T $\\
 $C_{77,T}= \omega(\lambda^2 , \lambda_y \lambda, \lambda^2 , \lambda_y \sigma)_0^T  $& 
 $C_{78,T}=  \omega(\lambda^2 , \lambda^2, \lambda_y \lambda , \lambda_y \sigma)_0^T $& 
 $C_{79,T}= \omega(\lambda^2 , \lambda_y \lambda, \lambda \sigma, \lambda_y \sigma)_0^T $\\
 $C_{80,T}= \omega(\lambda^2 , \lambda_y \sigma, \lambda^2 , \lambda_y \sigma)_0^T $& 
 $C_{81,T}= \omega(\lambda \sigma , \lambda_y \lambda , \lambda^2 , \lambda_y \sigma)_0^T $& 
 $C_{82,T}= \omega(\sigma^2 , \lambda \sigma_z , \lambda^2 , \lambda_y \sigma)_0^T $\\
 $C_{83,T}=\omega(\lambda^2 , \lambda^2, \lambda_y \sigma , \lambda_y \sigma)_0^T $& 
 $C_{84,T}=  \omega(\lambda^2 , \lambda \sigma, \lambda_y \lambda , \lambda_y \sigma)_0^T $& 
 $C_{85,T}= \omega(\lambda \sigma , \lambda^2, \lambda_y  \lambda , \lambda_y \sigma)_0^T $\\
 $C_{86,T}= \omega(\lambda^2 , \sigma^2, \lambda \sigma_z , \lambda_y \sigma)_0^T $& 
 $C_{87,T}= \omega( \sigma^2, \lambda^2, \lambda \sigma_z , \lambda_y \sigma)_0^T$  & \\
\hline
\end{tabular}\\[.35em]
  \caption{Weight coefficients involving 4 multiple integrals } \label{TabforCcoeff}
\end{table}

\subsection*{Acknowledgements}

Philip Ngare gratefully acknowledges the financial support from an IMU Berlin Einstein Foundation Fellowship. 
The authors also thank the seminar participants at Global Derivatives for their fruitful comments.

\bibliography{bibliography}
\bibliographystyle{abbrvnat}

\end{document}